\documentclass[letterpaper,11pt]{article}

\usepackage{xspace}
\usepackage{amsmath}
\usepackage{amsthm}
\usepackage{amssymb}
\usepackage{amsfonts}
\usepackage{typearea}
\typearea{13}

\usepackage[compact]{titlesec}

\newcommand{\sketch}[1]{}

\newtheorem{thm}{Theorem}[section]

\newtheorem{obs}[thm]{Observation}
\newtheorem{lem}[thm]{Lemma}
\newtheorem{lemma}[thm]{Lemma}
\newtheorem{cor}[thm]{Corollary}

\newtheorem{assumption}[thm]{Assumption}
\newtheorem{claim}[thm]{Claim}
\newtheorem{rem}[thm]{Remark}

\def \RR   {{\mathbb R}}

\def \OPT  {\mbox{\rm OPT}}

\date{}

\newcommand{\ox}{\ensuremath{\overline{x}}\xspace}
\newcommand{\oy}{\ensuremath{\overline{y}}\xspace}
\newcommand{\trans}{\intercal}

\newcommand{\removefromfile}[1]{}
\newcommand{\ignore}[1]{}

\usepackage{setspace}
\makeatletter
 \allowdisplaybreaks
\makeatother

\newenvironment{proofof}[1]{

\noindent{\bf Proof of {#1}:}}
{\hfill$\square$

}

\def\E{\mathbb{E}}
\def\oy{\overline{y}}
\def\ox{\overline{x}}
\def\oF{\overline{F}}
\def\cS{\mathcal{S}}
\def\sse{\subseteq}

\begin{document}
\title{Online Packing and Covering Framework with Convex Objectives}
\date{}
\author{Niv Buchbinder\thanks{Statistics and Operations Research Dept., Tel Aviv University, Research supported in part by ISF grant
954/11 and by BSF grant 2010426.}
\and
Shahar Chen\thanks{Technion - Israel Institute of Technology, Haifa, Israel. Work
supported by ISF grant 954/11 and BSF grant 2010426.}
\and
Anupam Gupta\thanks{Computer Science Department, Carnegie Mellon
    University, Pittsburgh, PA 15213, USA. Research partly supported by
    NSF awards CCF-1016799 and CCF-1319811.}
\and
Viswanath Nagarajan\thanks{Department of Industrial and Operations
  Engineering, University of Michigan, Ann Arbor, MI 48109.}
\and
Joseph (Seffi) Naor$^\dagger$}

\maketitle
\begin{abstract}
  We consider online fractional covering problems with a convex
  objective, where the covering constraints arrive over time. Formally,
  we want to solve
  \[ \min\,\{ f(x) \mid Ax\ge \mathbf{1},\, x\ge 0\},\] where the
  objective function $f:\RR^n\rightarrow \RR$ is convex, and the
  constraint matrix $A_{m\times n}$ is non-negative. The rows of $A$
  arrive online over time, and we wish to maintain a feasible solution
  $x$ at all times while only increasing coordinates of $x$.  We also
  consider packing problems of the form
  \[ \max\,\{ c^\trans y - g(\mu) \mid A^\trans y \le \mu,\,
  y\ge 0\},\] where $g$ is a convex function. In the online setting,
  variables $y$ and columns of $A^\trans$ arrive over time, and we wish
  to maintain a non-decreasing solution $(y,\mu)$. These problems are dual to each other when $g = f^\star$ the Fenchel
  dual of $f$.

  \medskip We provide an online primal-dual framework for both classes
  of problems with competitive ratio depending on certain
  ``monotonicity'' and ``smoothness'' parameters of $f$; our results
  match or improve on guarantees for some special classes of
  functions $f$ considered previously.


  \medskip Using this fractional solver  with problem-dependent
  randomized rounding procedures, we obtain competitive algorithms for
  the following problems: online covering LPs minimizing $\ell_p$-norms
  of arbitrary packing constraints, set cover with multiple cost
  functions, capacity constrained facility location, capacitated
  multicast problem, set cover with set requests, and 
  profit maximization with non-separable production costs. Some of these
  results are new and others provide a unified view of previous results,
  with matching or slightly worse competitive ratios.
\end{abstract}

\newpage

\section{Introduction}

We consider the following class of fractional covering problems:
\begin{equation}\label{eq:cnx-cov}
  \min\,\{ f(x)\,\, :\,\, Ax\ge 1,\, x\ge 0\}.
\end{equation}
Above, $f:\RR^n\rightarrow \RR$ is a non-decreasing convex function and
$A_{m\times n}$ is non-negative. (Observe that we can transform the more
general constraints $Ax \geq b$ with all non-negative entries into this
form by scaling the constraints.) The covering constraints $a_i^\trans
x\ge 1$ arrive online over time, and must be satisfied upon arrival. We
want to design an online algorithm that maintains a feasible fractional
solution $x$, where $x$ is required to be non-decreasing over time.

We also consider the Fenchel dual of~\eqref{eq:cnx-cov} which is the
following packing problem:
\begin{equation}\label{eq:cnx-pack}
  \max\,\{ \mathbf{1}^\trans y - f^\star(\mu)\,\, :\,\, A^\trans  y \le
  \mu,\, y\ge 0\}.
\end{equation}
Here, the variables $y_i$ along with columns of $A^\trans$ (or,
alternatively, rows of $A$) arrive over time, and the Fenchel dual is
formally defined in~(\ref{eq:fenchel}); see, e.g.,~\cite{Rock} for
background and properties.
Let $d$ denote the {\em row sparsity} of the matrix $A$, i.e., the maximum number of non-zeroes in any row, and let $\nabla_\ell f(z)$ be the $\ell^{th}$ coordinate of the gradient of $f$ at point $z \in \RR^n$.

This paper gives an online primal-dual algorithm for this pair of convex
programs~\eqref{eq:cnx-cov} and~\eqref{eq:cnx-pack}. This extends the
widely-used online primal-dual framework for linear objective functions to the
convex case. The competitive ratio is given as the ratio between the primal and dual objective functions\footnote{However, for clarity of exposition we provide the ratio as
Dual/Primal and not vice versa.}. It depends on certain
``smoothness'' parameters of the function $f$. We provide two general
algorithms:
\begin{itemize}
\item In the first algorithm, the primal variables $x$ and dual variables $\mu$ are monotonically non-decreasing, while the dual variables $y$ are allowed to both increase and decrease over time. The competitive ratio of this algorithm is:
  \begin{equation}\label{intro:gen-apx}
    \frac{{\rm Dual}}{{\rm Primal}}\,\,\ge \,\, \max_{c > 0} \,\,
    \left[\min_{z} \left( \frac{1}{8\log (1+d ) }\min_{\ell=1}^{n}\left\{\frac{\nabla_\ell
          f(z)}{\nabla_\ell f(cz)}\right\}\right) -
    \max_{z} \left(\frac{z^\trans \nabla f(z) -
        f(z) }{f(c z)}\right)\right].
  \end{equation}


\item In the second algorithm, all variables---primal variables $x$
  as well as dual variables $y,\mu$---are required to be monotonically non-decreasing. The
  competitive ratio is slightly worse in this case, given by:
  \begin{equation}\label{intro:gen-apx-mon}
    \frac{{\rm Dual}}{{\rm Primal}}\,\,\ge \,\, \max_{c > 0} \,\,
    \left[\min_{z} \left( \frac{1}{2 \log (1+d\rho)}\min_{\ell=1}^{n}\left\{\frac{\nabla_\ell
          f(z)}{\nabla_\ell f(cz)}\right\}\right) -
    \max_{z} \left(\frac{z^\trans \nabla f(z) -
        f(z) }{f(c z)}\right)\right].
  \end{equation}
  Observe that the difference from~(\ref{intro:gen-apx}) is the additional
  parameter $\rho$, which is defined to be an upper bound on the maximum-to-minimum ratio of positive
  entries in any column of $A$.
\end{itemize}
The above expressions are difficult to parse because of their generality, so
the first special case of interest is that of linear objectives. In this
case $z^\trans \nabla f(z) = f(z)$, and also $\nabla f(z) = \nabla
f(cz)$, hence the competitive ratios are $O(\log d)$ for monotone
primals, and $O(\log (d\rho))$ for monotone primals and duals. Both of
these competitive ratios are known to be best
possible~\cite{BN-MOR,GN12-mor}.

\def\b{b}  

The applicability of our framework extends to a number of settings,
most of which have been studied before in different works. We now
outline some of these connections.
\begin{itemize}
\item {\em Mixed Covering and Packing LPs}. In this problem, covering constraints $Ax\ge 1$ arrive online. There are also $K$ ``packing constraints'' $\sum_{j=1}^n \b_{kj}\cdot x_j \le \lambda_k$, for $k\in [K]$, that are given up-front. The right hand sides $\lambda_k$ of these packing constraints are themselves variables, and the objective is to minimize the $\ell_p$-norm $(\sum_{k=1}^K \lambda_k^p)^{1/p}$ of the ``load vector'' $\lambda = (\lambda_1, \ldots, \lambda_K)$. All entries $a_{ij}$ and $\b_{kj}$ are non-negative.
    Clearly, the objective function is a monotonically non-decreasing convex function.


  We obtain an $O(p\log d)$-competitive algorithm for this problem,
  where $d\le n$ is the row-sparsity of matrix $A$. Prior to our work,
  \cite{ABFP13} gave an $O(\log K\cdot \log (d\kappa
  \gamma))$-competitive algorithm for the special case of $p=\log K$
  (corresponding to $\| \lambda \|_\infty$, the makespan of the
  loads); here $\gamma$ and $\kappa$ are the maximum-to-minimum ratio of the
  entries in the covering and packing constraints.

\item {\em Set Cover with Multiple Costs}.
Here the offline input is a
  collection of $n$ sets $\{S_j\}_{j=1}^n$ over a universe $U$, and
  $K$ different linear  cost functions $B_k:[n]\rightarrow \mathbb{R}_+$ for
  $k\in [K]$. Elements from $U$ arrive online and must be covered by
  some set upon arrival, where the decision to select a set into the
  solution is irrevocable. The goal is to maintain a set-cover that
  minimizes the $\ell_p$ norm of the $K$ cost functions. Combining our
  framework with a simple randomized rounding scheme gives an
  $O(\frac{p^3}{\log p}\log d\log |U|)$-competitive
  randomized online
  algorithm; here $d$ is the maximum number of sets containing any
  element. The special case of $K=1$ (when $p=1$ without loss of generality)
  is the online set-cover problem~\cite{AAABN03}, for which the
  resulting $O(\log d \log |U|)$-competitive bound is tight, at least for
  randomized polynomial-time online algorithms~\cite{Korman05}.

\item {\em Capacity Constrained Facility Location} (CCFL). Here we are
  given $m$ potential facility locations, each with an opening cost $c_i$
  and a capacity $u_i$. Now, $n$ clients arrive online, each client $j\in
  [n]$ having an assignment cost $a_{ij}$ and a demand/load $\b_{ij}$ for
  each facility $i\in[m]$. The online algorithm must open facilities
  (paying the opening costs $c_i$) and assign each arriving client $j$
  to some open facility $i$ (paying the assignment cost $a_{ij}$, and
  incurring a load $p_{ij}$ on facility $i$). The \emph{makespan} of an
  assignment is the maximum load on any facility. The objective in CCFL
  is to minimize the sum of opening costs, assignment costs and the
  makespan. Using our framework, we obtain an $O(\log^2m)$-competitive
  fractional solution to a convex relaxation of CCFL. This is then
  rounded online to get an $O(\log^2m\,\log mn)$-competitive randomized
  online algorithm. This competitive ratio is worse by a logarithmic
  factor than the best result~\cite{ABFP13}, but it follows easily from
  our general framework.

\item {\em Capacitated Multicast Problem} (CMC).  This is a common
  generalization of CCFL and the \emph{online multicast
    problem}~\cite{AAABN-talg06}. There are $m$ edge-disjoint rooted
  trees $T_1,\cdots,T_m$ corresponding to multicast trees in some
  network. Each tree $T_i$ has a {\em capacity} $u_i$, and each edge
  $e\in \cup_{i=1}^m T_i$ has an opening cost $c_e$. A sequence of $n$
  clients arrive online, and each must be assigned to one of these
  trees. Each client $j$ has a tree-dependent load of $p_{ij}$ for tree
  $T_i$, and is connected to exactly one vertex $\pi_{ij}$ in tree
  $T_i$. Thus, if client $j$ is assigned to tree $T_i$ then the load of
  $T_i$ increases by $p_{ij}$, and all edges on the path in $T_i$ from
  $\pi_{ij}$ to its root must be opened. The objective is to minimize
  the total cost of opening the edges, subject to the capacity
  constraints that the total load on tree $T_i$ is at most $u_i$.
  Solving a natural fractional convex relaxation, and then applying a
  suitable randomized rounding to it, we get an $O(\log^2m\,\log
  mn)$-competitive randomized online algorithm that violates each
  capacity by an $O((d+\log^2m)\log mn)$ factor; here $d$ is the maximum
  depth of the trees $\{T_i\}_{i=1}^m$. The capacitated multicast
  problem with depth $d=2$ trees generalizes the CCFL problem, in which
  case we recover the above result for CCFL.

\item {\em Online Set Cover with Set Requests} (SCSR).  We are given a
  universe $U$ of $n$ \emph{resources}, and a collection of $m$
  \emph{facilities}, where each facility $i\in[m]$ is specified by (i) a
  subset $S_i\sse U$ of resources (ii) opening cost $c_i$ and (iii)
  capacity $u_i$. The resources and facilities are given up-front. Now,
  a sequence of $k$ {\em requests} arrive over time. Each request
  $j\in[k]$ requires some subset $R_j\sse U$ of resources. The request
  has to be served by assigning it to some collection $F_j \sse [m]$ of
  facilities whose sets collectively cover $R_j$, i.e., $R_j\sse
  \cup_{i\in F_j} S_i$. Note that these facilities have to be open, and
  we incur the cost of these facilities. Moreover, if a facility $i$ is
  used to serve client $j$, this contributes to the load of facility
  $i$, and this total load must be at most the capacity $u_i$.  This
  problem was considered recently by Bhawalkar et al.~\cite{BGP14}.

  Using an approach identical to that for the CCFL problem, we get an
  $O(\log^2m\,\log mnk)$-competitive randomized online algorithm that
  violates each capacity by an $O(\log^2m\,\log mnk)$ factor. Again this
  factor is weaker than the best result by a logarithmic factor, but
  directly follows from our general framework.

\item {\em Profit Maximization with Production Costs} (PMPC).
This is an application of the dual packing problem~\eqref{eq:cnx-pack}, in contrast to the above applications which are all applications of the primal covering problem.

  Consider a seller with $m$ items that can be produced and sold.  The seller has a \emph{production cost function} $g:\RR_+^m\rightarrow \RR_+$ which is monotone, convex and satisfies some other technical conditions; the total cost incurred by the seller to produce $\mu_j$ units of every item $j\in [m]$ is given by $g(\mu)$.\footnote{An important difference from prior work on such problems~\cite{BGMS11,HK15}: in these works, each item $j$ had a separate production cost function $g_j(\mu_j)$, and $g(\mu) := \sum_j g_j(\mu_j)$. We call this the \emph{separable} case. Our techniques allow the production cost to be {\em non-separable} over items---e.g., we can handle $g(\mu)=(\sum_{j=1}^m \mu_j)^2$.}  There are $n$ buyers who arrive online. Each buyer $i\in[n]$ is interested in subsets of items (bundles) that belong to a set family $\cS_i \sse 2^{[m]}$. The value of buyer $i$ for subset $S \in \cS_i$ is given by $v_i(S)$, where $v_i:\cS_i\rightarrow \RR_+$ is her \emph{valuation function}. If buyer $i$ is allocated a bundle $T\in \cS_i$, she pays the seller her valuation $v_i(T)$. (Observe: this is not an auction setting.) The goal in the PMPC problem is to produce items and allocate subsets to buyers so as to maximize the profit $\sum_{i=1}^n v_i(T_i) - g(\mu)$, where $T_i\in \cS_i$ denotes the subset allocated to buyer $i$ and $\mu\in \RR^m$ is the total quantity of all items produced.  As mentioned above, we consider a non-strategic setting, where the valuation of each buyer is known to the seller.

  Our main result here is for the fractional version of the problem where the allocation to each buyer $i$ is allowed to be any point in the convex hull of the set family $\cS_i$. We show that for a large class of valuation functions (e.g., supermodular, or weighted rank-functions of matroids) and production cost functions, our framework provides a polynomial time online algorithm: the precise competitive ratio is given by expression ~\eqref{intro:gen-apx-mon} with $f=g^\star$. As a concrete example, suppose the production cost function is $g(\mu)=( \sum_{j=1}^m \mu_j)^p$ for some $p > 1$. In this case, we get an $O(q\log \beta)^q$-competitive algorithm, where $q>1$ satisfies $\frac1q+\frac1p=1$, and $\beta$ is the maximum-to-minimum ratio of the valuation functions $\{v_i\}$.
\end{itemize}

As the above list indicates, the framework to solve fractional convex
programs is fairly versatile and gives good fractional results for a
variety of problems. In some cases, solving the particular relaxation we
consider and then rounding ends up being weaker than the best known
results for that specific problems (by a logarithmic factor); we hope
that further investigation into this problem will help close this gap.

\paragraph{Bibliographic Note:} In independent and concurrent work,
Azar et al.~\cite{ACP14} consider online covering problems with convex
objectives---i.e., problem~\eqref{eq:cnx-cov}. They also obtain a
competitive ratio that depends on properties of the function $f$, but
their parameterization is somewhat different from ours. As an example,
for online covering LPs minimizing the $\ell_p$-norm of packing
constraints, they obtain an $O(p\log (d\kappa \gamma))$-competitive
algorithm, whereas we obtain a tighter $O(p\log d)$ ratio.

\subsection{Techniques and Paper Outline}
\label{sec:outline}

In \S\ref{sec:alg-description}, we give the first general algorithm for the
convex covering problem~(\ref{eq:cnx-cov}) maintaining monotone primal
variables (but allowing dual variables to decrease). The main
observation is simple, yet powerful: convex
optimization problems with a function $f$ can be reduced to linear optimization using the
gradient of the convex function $f$. In the process we end up also
giving a cleaner algorithm and proof for linear optimization problems
as well, significantly simplifying the previous algorithm
from~\cite{GN12-mor}.  The resulting algorithm performs multiplicative
increases on the primal variables; for the dual, it does an initial
increase followed by a linear decrease after some point.

In \S\ref{sec:monotone} we give the second general algorithm, which is
simpler. The primal updates are the same as above but we skip the dual
decreases. This results in a worse competitive ratio, but the loss is
necessary for any monotone primal-dual
algorithm~\cite{BN-MOR}.

In \S\ref{sec:applications} and \S~\ref{sec:profit-max} we deal with the
various applications of our framework. The high-level idea in all of these
is to suitably cast each application in the form of either ~\eqref{eq:cnx-cov}
or~\eqref{eq:cnx-pack}. All, but the applications in
\S\ref{sec:applications}, are for the convex covering problem~\eqref{eq:cnx-cov}. Some
comments on the main ideas to watch out for:
\begin{itemize}
\item For applications to combinatorial problems we have to define the
  convex relaxation with some care in order to avoid bad integrality
  gaps. Moreover, some of our convex relaxations are motivated by the
  particular constraints we want to enforce when
  subsequently rounding.

\item For some of the problems our convex relaxations have an
  exponential number of constraints. To get a polynomial running time,
  we use the natural ``separation oracle'' approach. Moreover, we relax
  the constraints by a constant factor, so that each call to the
  separation oracle gives us a ``big'' improvement, and hence there are
  only a few updates per request.

\item For capacity constrained facility location (in \S\ref{sec:CCFL}),
  capacitated multicast problem (in \S\ref{sec:multicast}), and set
  cover with set requests (in \S\ref{sec:sc-set-requests}), na\"{\i}ve
  randomized rounding is bad, and hence the rounding schemes introduces
  correlations between opening facilities and assigning clients. These
  correlations also motivate the specific convex relaxations we consider
  for the problems.
\end{itemize}

In \S\ref{sec:profit-max} we consider the problem of profit maximization
with production costs, which after some simplifications can be cast as a
convex packing program as in~\eqref{eq:cnx-pack}.  We want allocations
to be non-decreasing over time, so we use our second general primal-dual
algorithm, which maintains monotone solutions. We also show how this
problem can be solved efficiently for some special classes of valuation
functions: supermodular and matroid-rank-functions. This convex program
can also be (randomly) rounded online to get integral allocations with the
same multiplicative competitive ratio, but with an extra additive term. The
additive term depends only on the number $m$ of items and the cost
function $g$; in particular it does not depend on $n$, the number of
buyers. We note that such an additive loss is necessary for our approach
due to an integrality gap of the convex relaxation.

\subsection{Related Work}
\label{sec:related-work}

This paper adds to the body of work in online primal-dual algorithms;
see~\cite{BN-mono} for a survey of this area. This approach has been
applied successfully to a large class of online problems: set
cover~\cite{AAABN03}, graph connectivity and cuts~\cite{AAABN-talg06},
caching~\cite{BBN-focs07-paging}, auctions~\cite{HK15}, scheduling~\cite{DH14},
etc. Below we discuss in more detail only work that is directly
relevant to us.

Online packing and covering {\em linear programs} were first considered by
Buchbinder and Naor~\cite{BN-MOR}, where they obtained an $O(\log n)$-competitive algorithm for covering and an $O(\log
(n\frac{a_{max}}{a_{min}}))$-competitive algorithm for packing. The competitive ratio for covering linear
programs was improved to $O(\log d)$ by Gupta and
Nagarajan~\cite{GN12-mor}, where $d\le n$ is the maximum number of
non-zero entries in any row.

Azar, Bhaskar, Fleischer, and Panigrahi~\cite{ABFP13} gave the first
algorithm for online {\em mixed packing and covering} LPs, where the
packing constraints are given upfront and covering constraints arrive
online; the objective is to minimize the maximum violation of the
packing constraints. Their algorithm had a competitive ratio of $O(\log
K\cdot \log (d\kappa \gamma))$, where $K$ is the number of packing
constraints and $\gamma$ (resp. $\kappa$) denotes the maximum-to-minimum
ratio of entries in the covering (resp. packing) constraints. Using our
framework, this bound can be improved to $O(\log K\cdot \log d)$. This
is also best possible as shown in~\cite{ABFP13}.

The capacity constrained facility location problem was also introduced
by Azar, Bhaskar, Fleischer, and Panigrahi~\cite{ABFP13}, who gave an
$O(\log m\log mn)$-competitive algorithm. Our result for this
problem is worse by a log-factor, but has the advantage of following
directly from our general framework. Moreover, our approach can be
extended to the capacitated multicast problem, which is a generalization
of CCFL to multi-level facility costs. The online multicast problem
(without capacities) was considered by Alon et al.~\cite{AAABN-talg06}
where they obtained an $O(\log m\cdot \log n)$-competitive randomized
algorithm.

The online set cover problem with set requests was considered recently
by Bhawalkar, Gollapudi, and Panigrahi~\cite{BGP14} who obtained an
$O(\log m\log mnk)$-competitive algorithm where capacities are violated
by an $O(\log^2 m\log mnk)$ factor. The competitive ratio obtained
through our approach is worse by a logarithmic factor in the cost
guarantee. Still, we think this is useful, since it follows with almost
no additional effort, given our online fractional framework and the CCFL
rounding scheme. Our approach is also likely to be useful in other such
generalizations.

The class of online maximization problems with production costs was
introduced by Blum, Gupta, Mansour, and Sharma~\cite{BGMS11} and
extended by Huang and Kim~\cite{HK15}. The key differences from our
setting are: (i) these papers deal with an {\em auction} setting where
the seller is not aware of the valuations of the buyers, whereas our
setting is not strategic, and (ii) these papers are restricted to
separable production costs, whereas we can handle much more general
(non-separable) cost functions.

\section{The General Framework}
\label{sec:general}

Let $f:\RR^n\rightarrow \RR$ be a non-negative non-decreasing convex
function. We assume that the function $f$ is continuous and
differentiable, and satisfies the following \emph{monotonicity
  condition}:
\begin{equation}
  \forall x \geq x' \in \RR^n,  \qquad \nabla f(x) \geq \nabla
  f(x') \label{prop-grad}
\end{equation}
 Here, $x \geq x'$ means $x_i \geq x'_i$ for all $i \in [n]$.

We consider the online fractional covering problem~\eqref{eq:cnx-cov}
where the constraints in $A$ arrive online. Our algorithm is a
primal-dual algorithm, which works with the following pair of convex
programs:
$$\begin{array}{lll|lll}
(P): &  \min &f(x) \qquad\qquad &\qquad  (D): &  \max &\sum_{i=1}^{m} y_i - f^\star(\mu) \\
&&Ax\ge 1 & & & y^\trans A \leq \mu^\trans \\
&& x\ge 0.  & &&  y\ge 0.
\end{array}
$$
Here $f^\star$ is the Fenchel dual of $f$, which is defined
as
\begin{gather}
  f^\star(\mu) = \sup_z \{ \mu^\trans z - f(z)\}. \label{eq:fenchel}
\end{gather}
 (Observe that by
scaling the rows of $A$ appropriately, we can transform any convering LP
of the form $Ax \geq b$ into the form above.) The following duality is
standard.

\begin{lemma}[Weak duality]
Let $x, (y,\mu)$ be feasible primal and dual solutions to $(P)$ and $(D)$ respectively. Then,
\begin{equation}
\mbox{\rm Primal objective }= f(x) \geq \sum_{i=1}^{m}y_i - f^\star(\mu)= \mbox{\rm Dual objective}.
\end{equation}
\end{lemma}
\begin{proof}
$$\sum_{i=1}^m y_i =   y^\trans \mathbf{1} \,\,  \le  \,\, y^\trans Ax  \,\, \le  \,\, \mu^\trans  x  \,\, =  \,\, \left( \mu^\trans  x - f(x)\right) + f(x) \,\, \le  \,\, f^\star(\mu) + f(x).$$
Rearranging we get the desired.
\end{proof}

\subsection{The Algorithm}
\label{sec:alg-description}

The algorithm maintains a feasible primal $x$ and a feasible dual solution $y$ at each time.

\vspace{0.25cm}
\noindent \framebox[1.05\width][l]{
\begin{minipage}{0.94\linewidth}
\vspace{0.1in}
{\bf Fractional Algorithm:} At round $t$:
\begin{itemize}
\item Let $\tau$ be a continuous variable denoting the current time.
\item While the new constraint is unsatisfied, i.e., $\sum_{j=1}^{n}a_{t j} x_j<1$, increase $\tau$ at rate $1$ and:
\item {\bf Change of primal variables:}
\begin{itemize}

\item For each $j$ with $a_{t j}>0$, increase each $x_j$ at rate
\begin{gather}
    \frac{\partial x_j}{\partial \tau}=  \frac{a_{t j}\,x_j + \frac{1}{d}}{\nabla_j f(x)}. \label{eq:1}
  \end{gather}
Here $d$ is an upper bound on the row sparsity of the matrix. $\nabla_j f(x)$ is the $j^{th}$-coordinate of the gradient $\nabla f(x)$.
\end{itemize}
\item {\bf Change in dual variables:}
\begin{itemize}
\item Set $\mu  = \nabla f(\delta x)$, where $\delta>0$ is determined later.
\item Increase $y_t$ at rate $r = \frac{1}{\log\left(1+ 2d^2\right)}\cdot \min_{\ell=1}^{n}\left\{\frac{\nabla_\ell f(\delta x)}{\nabla_\ell f(x)}\right\}$.
\item  If the dual constraint of variable $x_j$ is tight, that is, $\sum_{i=1}^{t} a_{ij}y_i = \mu_j$, then,
    \begin{itemize}
    \item Let $m^\star_j = \arg\max_{i=1}^{t}\{a_{ij} | y_i >0\}$.
    \item Increase $y_{m^\star_j}$ at rate $-\frac{a_{t j}}{a_{m^\star_j j}}\cdot r$.
\\(Note that this change occurs only if $a_{t j}$ is strictly positive.)
    \end{itemize}
\end{itemize}
\end{itemize}
\vspace{0.05in}
\end{minipage}}
\vspace{0.35cm}

We emphasize that the primal algorithm does not depend on the value $\delta$. The last step in the algorithm decreases certain dual variables; all other steps only increase primal and dual variables. For the analysis, we denote $x^{\tau},y^{\tau},\mu^{\tau},r^{\tau}$ as the value of $x,y,\mu,r$ at time $\tau$, respectively.

\begin{obs}
For any $\delta>0$, the following are maintained.
\begin{itemize}
\item The algorithm maintains a feasible monotonically non-decreasing primal solution.
\item The algorithm maintains a feasible dual solution with non-decreasing $\mu_j$.
\end{itemize}
\end{obs}

\begin{proof}
The first property follows by construction, since we only increase $x$ till reaching a feasible solution.
For the second property, we observe that the dual variables $\mu$ are non-decreasing since $\nabla f(x)$ is non-decreasing.
We prove that $y,\mu$ is feasible  by induction over the execution of the algorithm.
While processing constraint $t$, if $\sum_{i=1}^{t}a_{ij}y^{\tau}_i < \mu^{\tau}_j$ for column $j$ we are trivially satisfied. Suppose that during the processing of constraint $t$, we have $\sum_{i=1}^{t}a_{ij}y^{\tau}_i = \mu^{\tau}_j$ for some dual constraint $j$ and time $\tau$. Now the dual decrease part of the algorithm kicks in, and the rate of change in the left-hand side of the dual constraint is:
\[ \frac{d}{d \tau} \left(\sum_{i=1}^{t}a_{ij}y^{\tau}_i\right) = a_{t j} \cdot r^{\tau} - a_{m^\star_j j} \cdot \frac{a_{t j}}{a_{m^\star_j j}}\cdot r^{\tau}=0
 \]
\end{proof}

Before analyzing the competitive factor, let us first prove the following claim.

\begin{claim}\label{x_lower}
For a variable $x_j$, let $T_j=\{i | a_{ij}>0\}$ and let $S_j$ be any subset of $T_j$. Then,
\begin{equation}x^{\tau}_j \geq \frac{1}{\max_{i\in S_j}\{a_{ij}\}\cdot d}\left(\exp\left(\frac{\ln\left(1+2d^2\right)}{\mu_j^{\tau}}\sum_{i\in S_j}a_{ij}y^{\tau}_i\right)-1\right)\end{equation}
\end{claim}

\begin{proof}
Let $\tau(i)$ denote the value of $\tau$ at the arrival of the $i$th primal constraint. We first note that the increase in the primal variables at any time $\tau(i) \leq \tau \leq \tau(i+1)$ can be alternatively formulated by the following differential equation.
\begin{equation} \label{differential:xy}
\frac{\partial x_j}{\partial y_i}=  \frac{\log\left(1+ 2d^2\right)}{\min_{\ell=1}^{n}\left\{\frac{\nabla_\ell f(\delta x)}{\nabla_\ell f(x)}\right\}}\cdot\frac{a_{i j}\,x_j + \frac{1}{d}}{\nabla_j f(x)} \geq \log\left(1+ 2d^2\right)\cdot\frac{a_{i j}\,x_j + \frac{1}{d}}{\nabla_j f(\delta x)}.
\end{equation}
By solving the latter equation we get for any $\tau(i) \leq \tau \leq \tau(i+1)$,
\begin{eqnarray}
\frac{x_j^{\tau} + \frac{1}{a_{ij}d}}{x_j^{\tau(i)} + \frac{1}{a_{ij}d}} & \geq & \exp\left(\frac{\ln\left(1+2d^2\right)}{\nabla_j f(\delta x^\tau)}\cdot a_{ij}y^{\tau}_i\right) \label{x_lower:ineq1},
\end{eqnarray}
where we use the fact that $\nabla_j f(\delta x)$ is monotonically non-decreasing. Note that Inequality (\ref{x_lower:ineq1}) is satisfied even when no decrease is performed on the dual variables, and such a decrease only effects the right handside of the inequality.
For convenience, let us denote $\tau(t+1)=\tau$ (the actual value of $\tau(t+1)$ has not been revealed by the algorithm yet). Multiplying over all indices in $S_j$ we get,
\begin{eqnarray}
\lefteqn{\exp\left(\frac{\ln\left(1+2d^2\right)}{\mu_j^{\tau}} \sum_{i\in S_j} a_{ij}y^{\tau}_i\right) \leq \exp\left( \sum_{i\in S_j} \frac{\ln\left(1+2d^2\right)}{\nabla_j f(\delta x^{\tau(i+1)})} \cdot a_{ij}y^{\tau(i+1)}_i\right)} \label{x_lower:ineq2}\\
& \leq & \prod_{i\in S_j} \frac{x_j^{\tau(i+1)} + \frac{1}{a_{ij}d}}{x_j^{\tau(i)} + \frac{1}{a_{ij}d}}
  \leq   \prod_{i\in S_j} \frac{x_j^{\tau(i+1)} + \frac{1}{\max_{i\in S_j} \{a_{ij}\}\cdot d}}{x_j^{\tau(i)} + \frac{1}{\max_{i\in S_j} \{a_{ij}\}\cdot d}}\label{x_lower:ineq3} \\
& \leq & \prod_{i \in T_j} \frac{x_j^{\tau(i+1)} + \frac{1}{\max_{i\in S} \{a_{ij}\}\cdot d}}{x_j^{\tau(i)} + \frac{1}{\max_{i\in S_j} \{a_{ij}\}\cdot d}} = \frac{x_j^{\tau} + \frac{1}{\max_{i\in S_j} \{a_{ij}\}\cdot d}}{\frac{1}{\max_{i\in S_j} \{a_{ij}\}\cdot d}}\label{x_lower:ineq4} .
\end{eqnarray}
Inequality (\ref{x_lower:ineq2}) follows as $\mu_j^{\tau}=\nabla_j
f(\delta x^{\tau})$ and the value of $\nabla_j f(\delta x)$
monotonically non-decreases in time. Inequality (\ref{x_lower:ineq3})
follows by substituting (\ref{x_lower:ineq1}) into~(\ref{x_lower:ineq2}). Inequality (\ref{x_lower:ineq4}) follows as the value of $x_j$ monotonically non-decreases in time. Finally, the last equality is obtained using a telescopic sum and the fact that $x_j$ increases only in rounds with $a_{tj}>0$.
\end{proof}

\begin{thm}
\label{thm:cr-gen}
The competitive ratio of the algorithm is:
\begin{equation}
\min_{z} \left(\frac{\min_{\ell=1}^{n}\left\{\frac{\nabla_\ell f(\delta z)}{\nabla_\ell f(z)}\right\}}{4\ln(1+2d^2)}\right) -
\max_{z} \left(\frac{(\delta z)^\trans \nabla f(\delta z) - f(\delta z) }{f(z)}\right),
\end{equation}
where $\delta > 0$ is the parameter chosen in the algorithm.
\end{thm}

\begin{proof}
Consider the update when primal constraint $t$ arrives and $\tau$ is the current time.
Let $U(\tau)$ denote the set of tight dual constraints at time $\tau$. That is, for every $j\in U(\tau)$ we have $a_{tj}>0$ and $\sum_{i=1}^{t}a_{ij}y^{\tau}_i = \mu^{\tau}_j$. So $|U(\tau)|\le d$ the row-sparsity of $A$. Moreover, let us define for every $j\in U(\tau)$,  $S_j=\{i | a_{ij}>0, y^\tau_{i}>0\}$.
Clearly, $\sum_{i\in S_j}a_{ij}y^{\tau}_i = \sum_{i=1}^{t}a_{ij}y^{\tau}_i = \mu^{\tau}_j$, hence by Claim \ref{x_lower} and the fact that $\sum_j a_{t j}x_{j}^{\tau} < 1$, we get for every $j\in U(\tau)$,
$$ \frac{1}{a_{t j}} > x_{j}^{\tau} \geq \frac{1}{\max_{i\in S_j}\{a_{ij}\}\cdot d}\left(\exp\left(\ln(1+2d^2)\right)-1\right) ,$$
and after simplifying we get $ \frac{a_{t j}}{a_{m^\star_j j}} = \frac{a_{t j}}{\max_{i\in S_j}\{a_{i j}\}} \leq \frac{1}{2d}$.
As a result, we can bound the rate of change in the dual expression $\sum_{i=1}^{t}y_i$ at any time $\tau$:
\begin{align}
\frac{d \left(\sum_{i=1}^{t}y_i\right)}{d \tau} \geq r^{\tau} - \sum_{j \in U(\tau)}\frac{a_{t j}}{a_{m^\star_j j}} \cdot r^{\tau} \geq r^{\tau} \left(1 - \sum_{j \in U(\tau)}\frac{1}{2d}\right) \geq \frac{1}{2}r^{\tau} \label{d_dual},
\end{align}
where the last inequality follows as $|U(\tau)|\leq d$.

On the other hand, when processing constraint $t$ during the execution of the algorithm, the rate of increase of the primal objective $f$ is:
\begin{align}
\frac{d f(x^{\tau})}{d\tau} & = \sum_{j} \nabla_j f(x^{\tau}) \frac{\partial x^{\tau}_j}{\partial \tau}  =  \sum_{j | a_{tj}>0}\nabla_j f(x^{\tau}) \left(\frac{a_{tj}x^{\tau}_j + \frac{1}{d}}{\nabla_j f(x^{\tau})}\right) =  \sum_{j | a_{tj}>0}\left( a_{tj}x^{\tau}_j + \frac{1}{d}\right) \leq 2 \label{d_primal}.
\end{align}
The final inequality uses the fact that the covering constraint is unsatisfied, and that $d$ is at least the number of non-zeroes in the vector $a_t$.
From (\ref{d_dual}) and (\ref{d_primal}) we can now bound the following primal-dual ratio:
\begin{align}
\frac{d \left(\sum_{i=1}^{t}y^{\tau}_i\right)}{d f(x^{\tau})} \geq \frac{r^{\tau}}{4} = \frac{\min_{\ell=1}^{n}\left\{\frac{\nabla_\ell f(\delta x^{\tau})}{\nabla_\ell f(x^{\tau})}\right\}}{4\ln\left(1+ 2d^2\right)} \label{ineq:dy_increase}.
\end{align}
Thus, if $\ox$ and $\oy$ are the final primal and dual solutions we get,
\begin{align}
\sum_{i=1}^{m}\oy_i \geq \frac{\min_{x'}\min_{\ell=1}^{n}\left\{\frac{\nabla_\ell f(\delta x')}{\nabla_\ell f(x')}\right\}}{4\ln\left(1+ 2d^2\right)}\cdot f(\ox) \label{ineq:y_increase}.
\end{align}

To complete the proof of Theorem~\ref{thm:cr-gen}, we use the following standard claim. 
\begin{claim}
  \label{clm:dual-fact}
  For any $a\in \mathbb{R}^n$, we have $f^\star(\nabla f(a)) = a^\trans \nabla f(a) - f(a)$.
\end{claim}
\begin{proof}
By definition, $f^\star(\nabla f(a)) = \sup_x \{ x^\trans  \nabla f(a) -
f(x) \}$. Note that $x^\trans  \nabla f(a) - f(x)$ is concave as a function of $x$. So a necessary and sufficient condition for optimality is:
$$\nabla_i f(x) \,= \, \nabla_i f(a),\quad \forall i\in[n].$$
Thus setting $x=a$, we have $f^\star(\nabla f(a)) = a^\trans \nabla f(a) - f(a)$.
\end{proof}

Finally, we can attain the competitive ratio by a simple application of Claim \ref{clm:dual-fact} and Inequality (\ref{ineq:y_increase}) to the definition of the dual. Indeed,
\begin{align*}
    \mbox{Dual} & = \sum_{i=1}^{m}y_i - f^\star(\mu) \geq  \left(\frac{\min_{x'}\min_{\ell=1}^{n}\left\{\frac{\nabla_\ell f(\delta x')}{\nabla_\ell f(x')}\right\}}{4\ln(1+2d^2)} - \frac{f^\star(\nabla f(\delta\ox)) }{f(\ox)}\right) \cdot f(\ox) \\
\intertext{by Inequality~(\ref{ineq:y_increase}), and using Claim~\ref{clm:dual-fact} (with $a=\delta \ox$), we get}
    &= \left(\frac{\min_{x'}\min_{\ell=1}^{n}\left\{\frac{\nabla_\ell f(\delta x')}{\nabla_\ell f(x')}\right\}}{4\ln(1+2d^2)} - \frac{(\delta\ox)^\trans\nabla f(\delta \ox) - f(\delta \ox) }{f(\ox)}\right) \cdot f(\ox)
\\
    & \geq \left[\min_{z} \left(\frac{\min_{\ell=1}^{n}\left\{\frac{\nabla_\ell f(\delta z)}{\nabla_\ell f(z)}\right\}}{4\ln(1+2d^2)}\right) - \max_{z}\left(\frac{(\delta z)^\trans \nabla f(\delta z) - f(\delta z) }{f(z)}\right)\right] \cdot \mbox{Primal}
  \end{align*}
  Hence the proof.
\end{proof}

How to choose the value of $\delta$? If we  set $c = 1/\delta$ and optimize over $c$, the competitive ratio is:
\begin{equation}\label{eq:gen-apx}
  \frac{{\rm Dual}}{{\rm Primal}}\,\,\ge \,\, \max_{c > 0} \,\,
  \left(\min_{z} \frac{\min_{\ell=1}^{n}\left\{\frac{\nabla_\ell f(z)}{\nabla_\ell f(cz)}\right\}}{4\ln(1+2d^2)} -
  \max_{z} \frac{z^\trans \nabla f(z) - f(z) }{f(c z)}\right).
\end{equation}
This expression looks quite formidable, however it simply captures how
sharply the function $f$ changes locally. For special cases it gives
us very simple expressions; e.g., for linear cost functions $f(x) =
c^\trans x$ it gives us $Dual \geq Primal/O(\log d)$. See
\S\ref{sec:applications} for several such examples of applications using
this framework.

\subsubsection{Online Minimization}
\label{sec:minimization}

In the general framework above, we maintained both the primal and dual
solutions simultaneously. If our goal is to solve~(\ref{eq:cnx-cov})
online, i.e., to minimize the convex function $f(x)$ subject to covering
constraints arriving online, then the dual values can be determined with
hindsight once the final value of the primal variables $\ox$ has been
computed. In particular, we set $\mu = \nabla f(\delta \ox)$ once and
for all, and increase $y$ at a constant rate
\begin{gather*}
  \overline{r} = \frac{\min_{\ell=1}^{n}\left\{\frac{\nabla_\ell
        f(\delta \ox)}{\nabla_\ell f(\ox)}\right\}}{\log\left(1+
      2d^2\right)}.
\end{gather*}
These modifications can be easily plugged into the analysis above,
allowing us to omit the minimization over $x'$ in the competitive ratio.
(Observe that the update for the primal variables remains the same).

\begin{cor}
For online minimization, the competitive ratio of the algorithm is:
\begin{equation}
\max_{c>0} \quad \min_{z} \left(\frac{\min_{\ell=1}^{n}\left\{\frac{\nabla_\ell f( z)}{\nabla_\ell f(cz)}\right\}}{4\ln(1+2d^2)} - \frac{ z^\trans \nabla f( z) - f( z) }{f(cz)}\right)
\end{equation}
\end{cor}

\subsection{Monotone Online Maximization}
\label{sec:monotone}

If our goal is to solve~(\ref{eq:cnx-pack}) and maximize a dual
objective function subject to packing constraints, then indeed the above
framework increases the dual variables $\mu$, however the dual variables
$y$ can both increase and decrease. (Moreover, this potential decrease is
essential for the competitive ratio to be independent of the magnitude
of entries in the matrix $A$~\cite{GN12-mor}). In settings where
decrease in dual variables is not allowed, we need to slightly modify
(and simplify) the online dual update in the algorithm by setting
\begin{gather*}
  r = \frac{\min_{\ell=1}^{n}\left\{\frac{\nabla_\ell f(\delta
        x)}{\nabla_\ell f(x)}\right\}}{\log\left(1+ d\rho\right)},
\end{gather*}
where $\rho$ is an upper bound on $\frac{\max_t \{a_{tj}\}}{\min_{t,
    a_{tj}>0}\{a_{tj}\}} $ for all $1\leq j \leq n$.  And we skip the last step which decreases duals.
        Here,
application of Claim \ref{x_lower} at any round $t$ and time $\tau(t)
\leq \tau \leq \tau(t+1)$ yields
\begin{equation}
  \frac{1}{a_{tj}} \quad \geq \quad x^{\tau}_j \quad \geq\quad
  \frac{1}{\max_{i=1}^{t}\{a_{ij}\}\cdot d}
  \left(\exp\left(\frac{\ln\left(1 +
          d\rho\right)}{\mu^{\tau}_j} \sum_{i=1}^{t}a_{ij}y_i\right) -
    1\right),
\end{equation}
which implies $\ln\left(1+d\cdot\frac{\max_{i=1}^{t}\{a_{ij}\}}{a_{tj}}\right) / \ln\left(1+d\rho\right) \geq \frac{\sum_{i=1}^{t}a_{ij}y_i}{\mu^{\tau}_j}$, and thus guarantees $\sum_{i=1}^{t}a_{ij}y_i \leq \mu^{\tau}_j$.

\begin{cor}
For online maximization, when decreasing dual variables is not allowed, the adjusted  algorithm obtains the following competitive ratio:
\begin{equation}
\max_{c>0}\quad \min_{z} \left(\frac{\min_{\ell=1}^{n}\left\{\frac{\nabla_\ell f( z)}{\nabla_\ell f(cz)}\right\}}{2\ln(1+\rho d)}\right) - \max_{z} \left(\frac{ z^\trans \nabla f( z) - f( z)}{f(cz)}\right) \label{eq:mono-cr}
\end{equation}
\end{cor}

This results in a worse competitive ratio, but having monotone duals is
useful for two reasons: (a) in some settings we need monotone duals, as
in the profit maximization application in Section~\ref{sec:profit-max}), and
 (b) we get a simpler algorithm since we skip the third step of the
online dual update (involving the dual decrease).

\section{Applications}
\label{sec:applications}

We show how the general framework above can be used to give algorithms
for several previously-studied as well as new problems. In contrast to previous papers
where a primal-dual algorithm had to be tailored to each of these
problems, we use the framework above to solve the underlying convex
program, and then apply a suitable rounding algorithm to the fractional
solution.

\subsection{$\ell_p$-norm of Packing Constraints}
\label{sec:Lp-norm-packing}

We consider the problem of solving a mixed packing-covering linear
program online, as defined by Azar et al.~\cite{ABFP13}. The covering
constraints $Ax\ge 1$ arrive online, as in the above setting. There are
also $K$ ``packing constraints'' $\sum_{j=1}^n \b_{kj}\cdot x_j \le
\lambda_k$ for $k\in [K]$ that are given up-front. The right sides
$\lambda_k$ of these packing constraints are themselves variables, and
the objective is to minimize $\sum_{k=1}^K \lambda_k^p$  or alternatively, $\|\lambda\|_p=\sqrt[p]{\sum_{k=1}^K \lambda_k^p}$. All the entries
in the constraint matrices $A = (a_{ij})$ and $B=(\b_{kj})$ are
non-negative.

\begin{thm}
  \label{thm:p-norm}
  There is an $O(p\log d)$-competitive online algorithm for fractional
  covering with the objective of minimizing $\ell_p$-norm of multiple
  packing constraints.
\end{thm}

\begin{proof}
  In order to apply our framework to this problem, we seek to minimize
  the convex function
  \begin{gather*}
    f(x) ~~=~~ \frac1p \| Bx \|_p^p \,\,=\,\,\frac{1}{p}\sum_{k=1}^{K}\left(B_k
      x\right)^p \,\,=\,\, \frac{1}{p}\sum_{k=1}^{K}\left( \sum_{j=1}^n
      \b_{kj}\cdot x_j
    \right)^p .
  \end{gather*}
  This is the $p$-power of the original objective; above
  $B_k=(\b_{k1},\cdots,\b_{kn})$ is the $k^{th}$ packing constraint.

  To obtain the competitive ratio, observe that $\nabla_j f(x) =
  \sum_{k=1}^K \b_{kj}\cdot \left(P_k x\right)^{p-1}$.  Thus, we
  have for all $c>0$, $x\in \RR^n_+$ and $1\leq j \leq n$:
  \begin{align*}
    \frac{f(z)}{f(cz)} & = (1/c)^p \\
    \frac{\nabla_j f(z)}{\nabla_j f(c z)} &= (1/c)^{p-1} \\
    \frac{\sum_{j=1}^{n}z_j \cdot \nabla f(z)_j}{f(cz)} &=
    \frac{\sum_{j=1}^{n}z_j\sum_{k=1}^K \b_{kj }\cdot \left(B_k z\right)^{p-1}}{f(cz)}=     \frac{p\cdot f(z)}{f(cz)}  = p(1/c)^{p}.
     \end{align*}
Substituting $\delta = 1/c$ and plugging into~\eqref{eq:gen-apx} we get:
  \begin{equation}
    {\rm Dual} \geq \left(\frac{\delta^{p-1}}{4 \ln(1 + 2d^2)} - p \delta^p +
      \delta^p\right) \cdot {\rm Primal}
  \end{equation}
  So the primal-dual ratio (as a function of $\delta$) is
  $\frac{{\rm Dual}}{{\rm Primal}} \ge \delta^{p-1}/L - (p-1)\delta^p$ where
  $L= 4 \ln(1 + 2d^2)$. This quantity is maximized when $\delta=\frac{1}{pL}$,
  leading to a primal-dual ratio of $1/(pL)^p$. Taking the $p^{th}$ root
  of this quantity gives us that the $\ell_p$-norm of the primal is at
  most $pL = O(p \log d)$ times the optimum.
\end{proof}

When $p=\Theta(\log m)$, the $\ell_p$ and $\ell_\infty$ norms are within
constant factors of each other, we obtain the online mixed
packing-covering LP (OMPC) problem studied by Azar et al.~\cite{ABFP13}.
For this setting this gives an improved $O(\log d\cdot \log
m)$-competitive ratio, where $d$ is the row-sparsity of the matrix $A$,
and $m$ is the number of packing constraints. This competitive ratio is
known to be tight~\cite[Theorem~1.2]{ABFP13}.

\begin{rem}\label{rem:diff-p} The above result also holds if function $f$ is the sum of distinct powers of linear functions, i.e. $f(x) = \sum_{k=1}^K (B_k x)^{p_k}$ where $p_1,\cdots,p_K\ge 1$ may be non-uniform. For this case, we obtain an $O(p\log d)$-competitive algorithm where $p=\max_{k=1}^K p_k$.
\end{rem}

\subsection{Online Set Cover with Multiple Costs}
\label{sec:set-cover}

Consider the online set-cover problem~\cite{AAABN03} with $n$ sets
$\{S_j\}_{j=1}^n$ over some ground set $U$. Apart from the set system, we
are also given $K$ cost functions $B_k:[n]\rightarrow \mathbb{R}_+$ for
$k\in [K]$. Elements from $U$ arrive online and must be covered by some
set upon arrival; the decision to select a set into the solution is
irrevocable. The goal is to maintain a set-cover that minimizes the
$\ell_p$ norm of the $K$ cost functions. We use Theorem~\ref{thm:p-norm}
along with a rounding scheme (similar to~\cite{GKP11-waoa}) to obtain:
\begin{thm}
  \label{thm:p-norm:int}
  There is an $O\left(\frac{p^3}{\log p} \log d \log
    r\right)$-competitive randomized online algorithm for set cover
  minimizing the $\ell_p$-norm of multiple cost-functions. Here $d$ is
  the maximum number of sets containing any element, and $r=|U|$ is the
  number of elements.
\end{thm}

\begin{proof}
  We use the following convex relaxation. There is a variable $x_j$ for
  each set $j\in [n]$ which denotes whether this set is chosen.
  \begin{alignat*}{2}
    \min & \quad g(x) \,\,=\,\, \sum_{k=1}^{K}\bigg( \sum_{j=1}^n
    \b_{kj}\cdot x_j \bigg)^p
    \,+\, \sum_{j=1}^n \bigg(\sum_{k=1}^K \b_{kj}^p\bigg)\cdot x_j \\
    s.t. & \quad \sum_{j : e\in S_j} x_j  \geq 1, \qquad \forall e \in U\\
    &\quad x \geq 0.
  \end{alignat*}
  We can use our framework to solve this fractional convex covering
  problem online. Although the objective has a linear term in addition
  to the $p$-powers, we obtain an $O(p\log d)^p$-competitive algorithm as noted in Remark~\ref{rem:diff-p}.

Let $C^\star$ denote the $p^{th}$
  power of the optimal objective of the given set cover instance. Then
  it is clear that the optimal objective of the above fractional
  relaxation is at most $2C^\star$. Thus the objective of our fractional online
  solution $g(x)= O(p\log d)^p\cdot C^\star$.

\def\oe{{\overline{e}}}

  To get an integer solution, we use a simple online randomized rounding
  algorithm. For each set $j\in [n]$, define $X_j$ to be a
  $\{0,1\}$-random variable with $\Pr[X_j=1]=\min\{4p\log r\cdot x_j,\,
  1\}$. This can easily be implemented online. It is easy to see by a
  Chernoff bound that for each element $e$, it is not covered with
  probability at most $\frac{1}{r^{2p}}$. 
If an element $e$ is not covered by this rounding, we choose the set minimizing $\min_{j=1}^n \{\sum_{k=1}^K \b_{kj}^p : e\in S_j\}$; let $\oe\in [n]$ index this set and $C_\oe=\sum_{k=1}^K \b_{k\oe}^p$. Observe that $C_\oe\le C^\star$ for all $e\in U$.

  To bound the $\ell_p$-norm of the cost, let $C_k=\sum_{j=1}^n
  \b_{kj}\cdot X_j$ be the cost of the randomly rounded solution under the
  $k^{th}$ cost function, and let $C:=\sum_{k=1}^K C_k^p$. Also for each element $e\in U$, define:
\begin{itemize}
\item $D_{ek}=\b_{k\oe}$ for all $k\in[K]$ and $D_e=C_\oe$ if $e$ is not covered by the rounding.
\item  $D_{ek}=0$ for all $k\in[K]$ and $D_e=0$ otherwise.
\end{itemize}
Note that $D_e=\sum_{k=1}^K D_{ek}^p$.  The
  $p^{th}$ power of the objective function is:
\begin{eqnarray}
\overline{C} & =& \sum_{k=1}^K  \left(C_k +\sum_{e\in U} D_{ek}\right)^p \quad \le \quad 2^p \sum_{k=1}^K  C_k^p + 2^p \sum_{k=1}^K  \left(\sum_{e\in U} D_{ek}\right)^p \quad \le \quad  2^p \cdot C + 2^p \sum_{k=1}^K  r^p \sum_{e\in U} D_{ek}^p \notag \\
&=& 2^p \cdot C+ (2r)^p \sum_{e\in U} D_{e}  \label{eq:multSC-1}
\end{eqnarray}
We now bound $\E[\overline{C}]$ using~\eqref{eq:multSC-1}. Observe that $\E[C_k]\le
  4p\log r\cdot \sum_{j=1}^n \b_{kj}\cdot x_j$. Since each $C_k$ is the
  sum of independent non-negative random variables, we can bound $\E[C_k^p]$ using a
  concentration inequality involving $p^{th}$ moments~\cite{Latala}:
$$\E[C_k^p]  \le  K_p \cdot \bigg( \E[C_k]^p + \sum_{j=1}^n \E[\b_{kj}^p
\cdot X_j^p] \bigg) \le K_p \cdot \bigg( (4p\log r)^p \;
\bigg(\sum_{j=1}^n \b_{kj}\cdot x_j\bigg)^p + 4p\log r \sum_{j=1}^n
\b_{kj}^p \cdot x_j \bigg) .$$ Above $K_p=O(p/\log p)^p$. By linearity of
expectation,
$$\E[C] = \sum_{k=1}^K \E[C_k^p] \le K_p (4p\log r)^p \sum_{k=1}^K \bigg( \bigg(\sum_{j=1}^n \b_{kj}\cdot x_j\bigg)^p + \sum_{j=1}^n \b_{kj}^p \cdot x_j  \bigg)  = K_p (4p\log r)^p \cdot g(x).$$
Thus we have $\E[C] = O\left(\frac{p^3}{\log p}\cdot \log d\cdot \log
  r\right)^p\cdot C^\star$.

Observe that $\E\left[ \sum_{e\in U} D_{e}  \right] = \sum_{e\in U} \Pr[e\mbox{ uncovered}]\cdot C_\oe \le r^{-2p} \cdot \sum_{e\in U} C^\star = r^{1-2p}\cdot C^\star$. Using these bounds in~\eqref{eq:multSC-1}, we have $\E[\overline{C}] \le
2^p \cdot \E[C]+ (2r)^p \sum_{e\in U} \E[D_{e}]  = O\left(\frac{p^3}{\log p}\cdot \log d\cdot \log
  r\right)^p\cdot C^\star$. \end{proof}

\subsection{Capacity-constrained Facility Location (CCFL)}
\label{sec:CCFL}

In the Capacity-constrained Facility Location (CCFL) problem, there are
$m$ potential facility locations each with an opening cost $c_i$ and a
capacity $u_i$ that are given up-front. There are $n$ clients which
arrive online. Each client $j\in [n]$ has, for each facility $i\in[m]$,
an assignment cost $a_{ij}$ and a demand/load $p_{ij}$. The online
algorithm needs to open facilities (paying the opening costs) and assign
each arriving client $j$ to some open facility $i$ (paying the
assignment cost $a_{ij}$, and incurring a load $p_{ij}$ on $i$). The
\emph{makespan} of an assignment is the maximum load on any
facility. The objective in CCFL is to minimize the sum of opening costs,
assignment costs and the makespan. An integer programming formulation for
this problem is the following:
\begin{alignat*}{2}
  \min & \sum_{i=1}^{m}c_i x_i + \sum_{i,j}a_{ij}y_{ij} + \max_{i=1}^{m}\sum_{j=1}^{n}p_{ij} \cdot y_{ij}\\
  s.t. & \quad \sum_{i\in S} x_i + \sum_{i\notin S}y_{ij} \geq 1, \qquad
  \forall j \in [n], \forall S \subseteq [m] \\
  &\quad y, x \in \{0,1\}.
\end{alignat*}
In order to apply our framework to CCFL, we allow the variables to be
fractional, and use the following objective function with $p=\Theta(\log
m)$.
\[ f(x,y) \,\,=\,\, \bigg(\sum_{i=1}^{m}c_i x_i \bigg)^p \,+\, \bigg(
\sum_{i,j}a_{ij}y_{ij}\bigg)^p \, + \,
\sum_{i=1}^{m}\bigg(\sum_{j=1}^{n}p_{ij} \cdot y_{ij}\bigg)^p. \] Note
that $f(x,y)^{1/p}$ is within a constant factor of the original
objective. We refer to the above convex program as the fractional CCFL
problem.

\begin{thm}\label{thm:f-CCFL}
There is an $O(\log^2 m)$-competitive online algorithm for fractional CCFL.
\end{thm}

\begin{proof}
  To apply our framework to solving fractional CCFL, we need a few
  observations. Firstly, although the function $f$ is not fully known in
  advance, we know at any time the parts of $f$ that correspond to
  variables appearing in constraints revealed until then. It is easy to
  check that this suffices for our framework to apply.

  Another issue is that there are an exponential number of covering
  constraints. This does not affect the $O(p\log d) = O(\log^2 m)$
  competitive ratio we obtain through Theorem~\ref{thm:p-norm}, since it
  is independent of the number of covering constraints. However, the
  running time will be exponential in the straightforward
  implementation. In order to obtain a polynomial running time, we relax
  the covering constraints to $\frac12$ (instead of one). Upon arrival
  of client $j$, we add covering constraints based on the following
  procedure.
  \begin{quote}
    While there is some $S\subseteq [m]$ with $\left(\sum_{i\in S} x_i +
      \sum_{i\notin S}y_{ij} < \frac12\right)$, do: \\
    $~~~~~$Add constraint $\sum_{i\in S} x_i + \sum_{i\notin S}y_{ij}
    \ge
    1$, and \\
    $~~~~~~~~$update solution $(x,y)$ according to the algorithm of
    Theorem~\ref{thm:p-norm}.
  \end{quote}
  Note that given a current solution $(x,y)$, the set $S$ that minimizes
  $\sum_{i\in S} x_i + \sum_{i\notin S}y_{ij}$ is $S=\{i\in[m] \mid x_i
  < y_{ij}\}$; comparing this to $1/2$ gives us the desired separation
  oracle. The number of iterations of the new procedure (per client
  arrival) is at most $4m$, because $\sum_{i=1}^m
  \left(\min\{x_i,1\}+\min\{y_{ij},1\} \right)$ increases by at least
  $\frac12$ in each iteration, and this sum is always between $0$ and
  $2m$. Hence, at any time $(2 x, 2y)$ is a feasible fractional
  solution, which satisfies all constraints.
\end{proof}

\subsubsection{Rounding the Fractional Solution Online}
\label{sec:ccfl-round}

The online fractional solution can be rounded in an online fashion to
obtain a randomized $O(\log^2 m\cdot \log mn)$-competitive
algorithm. While this is worse by a $\log m$ factor than the result
in~\cite{ABFP13}, it follows directly from our general algorithm.

We use a ``guess and double'' approach in the rounding. Let $M$ denote
some upper bound on the optimal offline value. Upon arrival of a new
client, our algorithm will succeed if $M$ is a correct upper bound. If
the algorithm fails then $M$ is doubled and we repeat the updates. We
start with $M$ being some known lower bound. A {\em phase} is a sequence
of client arrivals for which $M$ remains the same. At any point in the
algorithm, the only allowed facilities are $\{i \in [m] \mid c_i\le
M\}$ and the only allowed assignments are $\{ (i,j) \mid i \in [m],
j\in[n], p_{ij}\le M\}$. We denote by $I_M$ the restricted instance
which consists only of the clients that arrive in this phase and the
above facilities and allowed assignments. When we progress from one
phase to the next (i.e. $M$ is doubled), we reset all the $x,y$
variables to zero.

Define a modified objective as follows:
$$g(x,y) \,\,=\,\, \bigg(\sum_{i}^{}c_i \bigg(x_i + \frac{\sum_j p_{ij} \cdot y_{ij}}{M}\bigg)\bigg)^p \,+\, \bigg( \sum_{i,j}a_{ij}y_{ij}\bigg)^p \, + \,  \sum_{i}^{}\bigg(\sum_{j}^{ }p_{ij} \cdot y_{ij}\bigg)^p.$$
Note that this depends on the guess $M$ and is fixed for a single
phase. Below we focus on the restricted instance $I_M$. Unless specified
otherwise, clients $j$ and facilities $i$ are only from $I_M$. Consider
the following convex program:
\begin{align}
  \min & \quad g(x,y) \label{eq:gLP} \\
  s.t. & \quad \sum_{i\in S} x_i + \sum_{i\notin S}y_{ij} \geq 1, \qquad
  \forall j\in [n], S \subseteq [m] \notag\\
  &\quad y, x \geq 0. \notag
\end{align}
When a new client $h$ arrives, the algorithm first updates the
fractional solution to ensure the covering-constraints of client $h$ up
to a factor $2$, as in Theorem~\ref{thm:f-CCFL}. Now we have to do the
rounding. To do this, first define the following modified variables:
\begin{gather*}
\overline{y}_{ij}=\min\{y_{ij},x_i\}, \quad \forall i,j, \text{ and} \\
\ox_i = \max\left\{ x_i, \frac{\sum_j p_{ij} \cdot y_{ij}}{M} \right\},\quad \forall i.
\end{gather*}

By construction, the variables $(\ox,\oy)$ clearly satisfy:
\begin{alignat}{2}
  \textstyle \sum_j p_{ij}\cdot \oy_{ij} &\le M\cdot \ox_i & \qquad\qquad &\forall
  i \label{eq:ccfl:1} \\
  \textstyle \sum_{i=1}^m \oy_{ij} &\ge \frac12    &  &\forall j
  \label{eq:ccfl:2} \\
  \oy_{ij} &\le \ox_i &  &\forall i, j
  \label{eq:ccfl:3}
\end{alignat}

\begin{claim}
  \label{cl:ccfl-LB}
  Suppose there exists an integral solution to the current CCFL instance
  having cost at most $M$. Then the following inequalities hold, where
  $\alpha=O(\log^2m)$ is the competitive ratio in
  Theorem~\ref{thm:f-CCFL}:
  \begin{alignat}{2}
    \sum_i c_i\cdot \ox_i &\le 4\alpha \cdot M & &
    \label{eq:ccfl:4} \\
    \sum_{i,j}a_{ij}\cdot \oy_{ij}  &\le 4\alpha\cdot M & &
    \label{eq:ccfl:5} \\
    \sum_j p_{ij}\cdot \oy_{ij} &\le 4\alpha \cdot M &\qquad\qquad &\forall i.
    \label{eq:ccfl:6}
  \end{alignat}
\end{claim}
\begin{proof}
  Since the optimal integral value of the current CCFL instance is at
  most $M$, the optimal CCFL value of the restricted instance $I_M$ is
  also at most $M$. That is, there is an integral assignment with
  opening cost $\le M$, assignment cost $\le M$, and maximum load $\le
  M$. So the optimal fractional value of program~(\ref{eq:gLP}) is at
  most $(2M)^p+M^p+m\cdot M^p\le m(3M)^p$. Since the fractional
  algorithm in Theorem~\ref{thm:f-CCFL} is $\alpha$-competitive, we
  have:
  \[ g(x,y)\le \alpha^p\cdot m(3M)^p \leq (4\alpha M)^p, \]
  since $m \leq (4/3)^p$ for $p \geq \log_{4/3} m$.
  This implies:
  \begin{gather*}
    \sum_i c_i\cdot \ox_i \le \sum_{i} c_i \left(x_i + \frac{\sum_j
        p_{ij} \cdot y_{ij}}{M}\right) \le g(x,y)^{1/p} \\
    \sum_{i,j}a_{ij}\cdot \oy_{ij} \le \sum_{i,j}a_{ij}y_{ij} \le
    g(x,y)^{1/p}\\
    \sum_j p_{ij} \cdot \oy_{ij} \le \sum_j p_{ij} \cdot y_{ij} \le
    g(x,y)^{1/p}
  \end{gather*}
  and $g(x,y)^{1/p} \leq 4\alpha M$ proves all three claims.
\end{proof}

Hence, after the fractional updates, we check whether the conditions
in~\eqref{eq:ccfl:4}-\eqref{eq:ccfl:6} are satisfied; if not, we end the
phase and double $M$ (knowing by Claim~\ref{cl:ccfl-LB} that $M$ is a
lower bound on the CCFL instance so far), and start the next phase with
the new client $h$ and the new value of $M$. So assume that after
fractionally assigning $h$, all the inequalities
\eqref{eq:ccfl:1}-\eqref{eq:ccfl:6} hold for the current value $M$. Now
we perform randomized rounding as follows.

\begin{itemize}
\item For each $i$, set $X_i$ to $1$ with probability $\min\{4\log (mn)
  \cdot \ox_i, 1\}$. Let $F_f =\{i : \ox_{i} \geq \frac{1}{4\log (mn)}\}$
  denote the set of {\em fixed} facilities for which $\Pr[X_i=1]=1$.

\item For each $i,j$, define $Z_{ij}$ as follows:
  \begin{gather*}
    \Pr[ Z_{ij}=1] =
      \begin{cases}
        \min\{4\log mn\cdot \oy_{ij}, 1\} & \mbox{ if }i\in F_f, \\
        \frac{\oy_{ij}}{\ox_i} & \mbox{ otherwise.}
      \end{cases}
  \end{gather*}
\end{itemize}
All the above random variables are independent. Each client $j$ is
assigned to some facility $i$ with $X_i \cdot Z_{ij}=1$; if there are
multiple possible assignments, the algorithm breaks ties arbitrarily.
(For the sake of analysis, we may imagine that the client is assigned to
\emph{all} facilities such that $X_i \cdot Z_{ij}=1$.) If client $j$ is unassigned, we open the facility
corresponding to $\min_{i=1}^m (c_i+a_{ij}+p_{ij})$ and assign $j$ to it (note that this minimum value is at most $M$); we will show that this event happens with low probability, so the effect on the objective will be small.
We now analyze this rounding.

\begin{claim}
  For any client $j$, $\Pr[j\mbox{ not assigned}] = \Pr[ \sum_i X_i\cdot
  Z_{ij}=0] <1/n^2$.
\end{claim}

\begin{proof}
  If $i\in F_f$, then $E[X_i Z_{ij}] =E[Z_{ij}] = \min\{4\log mn\cdot
  \oy_{ij},1\}$. Else, $E[X_i Z_{ij}] = 4\log mn \cdot \oy_{ij} \ge
  4\log n\cdot \oy_{ij}$. In either case,
  \begin{gather*}
    \Pr[j\mbox{ not assigned}] = \Pr[ \sum_i X_i Z_{ij}=0] = \prod_i (1-E[
    X_i Z_{ij}]) \le \exp\bigg(-4\log n\sum_i \oy_{ij}\bigg)<1/n^2,
  \end{gather*}
  where the last inequality is by~\eqref{eq:ccfl:2}.
\end{proof}

\begin{claim}
  For any facility $i\in F_f$, we have $\Pr[load > 32\alpha \log mn\cdot
  M]\le 1/m^2$.
\end{claim}
\begin{proof}
  For facility $i\in F_f$, the load assigned to it is $\sum_j p_{ij}\cdot
  Z_{ij}$. This is a sum of independent $[0,M]$-bounded random variables
  (by definition of the restricted instance $I_M$), with expectation at
  most $4\log mn \sum_j p_{ij}\cdot \oy_{ij}$, which
  by~\eqref{eq:ccfl:6} is at most $16\alpha\log mn\cdot M$. The claim
  now follows by a Chernoff bound.
\end{proof}

\begin{claim}
  For any facility $i\not\in F_f$, we have $\Pr[load > 4\log mn\cdot M
  \mid X_i=1]\le 1/m^2$.
\end{claim}
\begin{proof}
  Fix $i\not\in F_f$ and condition on $X_i=1$. The load assigned to $i$
  is $\sum_j p_{ij}\cdot (Z_{ij}|X_i=1)$, which is a sum of independent
  $[0,M]$-bounded random variables (again by definition of the
  restricted instance). The expectation is at most $\sum_j p_{ij}\cdot
  \frac{\oy_{ij}}{\ox_i} \le M$, by~\eqref{eq:ccfl:1}. The claim again
  follows by a Chernoff bound.
\end{proof}

\begin{claim}
  $\Pr[\mbox{opening cost }> 32\alpha \log mn\cdot M] < 1/n^2$.
\end{claim}
\begin{proof}
  The opening cost is $\sum_i c_i\cdot X_i$ which is a sum of
  independent $[0,M]$-bounded random variables, whose expectation is at
  most $4\log mn \sum_i c_{i}\cdot \ox_i \le 16\alpha \log mn \cdot M$
  by~\eqref{eq:ccfl:4}. The claim now follows by a Chernoff bound.
\end{proof}

\begin{claim}
  $E[\mbox{assignment cost}] \le  16\alpha\log mn \cdot M$.
\end{claim}
\begin{proof}
  The assignment cost is $\sum_i\sum_j a_{ij}\cdot X_iZ_{ij}$ which has
  mean at most $4\log mn \sum_{ij}a_{ij} \oy_{ij}\le 16\alpha\log mn
  \cdot M$ by \eqref{eq:ccfl:5}.
\end{proof}

Combining the above claims, and using the fact that each element is uncovered with probability less than $\frac1{n^2}$,  we get:
\begin{lem}\label{lem:i-ccfl}
  The expected sum of opening and assignment costs and makespan is
  $O(\alpha\log mn)\cdot M$.
\end{lem}

A standard doubling argument accounts for all the phases as follows. Let
$M^\star$ denote the final value of the parameter $M$ achieved by the
algorithm. By Claim~\ref{cl:ccfl-LB} we have $OPT>M^\star/2$. On the other
hand, the expected cost in any phase corresponding to $M$ is at most
$O(\alpha\log mn)\cdot M$ by Lemma~\ref{lem:i-ccfl}; this gives a
geometric sum with total cost at most $O(\alpha\log mn)\cdot
(M^\star+\frac{M^\star}{2}+\cdots )\le O(\alpha\log mn)\cdot OPT$. This proves
the following theorem.
\begin{thm}
  There is a randomized $O(\log^2m \log mn)$-competitive ratio for CCFL.
\end{thm}

\begin{rem}
  We can use randomized rounding with alteration, as in~\cite{GN12-mor},
  to obtain a more nuanced $O(\log^2 m\cdot \log m\ell)$-competitive
  ratio, where $\ell\le n$ is the ``machine degree'' i.e. $\max_{i\in
    [m]} |\{j : p_{ij}<\infty\}|$. We omit the details.
\end{rem}


\subsection{Capacitated Multicast Problem}
\label{sec:multicast}

We consider the online multicast problem~\cite{AAABN-talg06} in the
presence of capacities, which we call the \emph{Capacitated Multicast}
(CMC) problem. In this problem, there are $m$ edge-disjoint rooted trees
$T_1,\cdots,T_m$ corresponding to multicast trees in some network. Each
tree $T_i$ has a {\em capacity} $u_i$ which is the maximum load that can
be assigned to it. Each edge $e\in \cup_{i=1}^m T_i$ has an opening cost
$c_e$. A sequence of $n$ clients arrive online, and each must be
assigned to one of these trees. Each client $j$ has a tree-dependent
load of $p_{ij}$ for tree $T_i$, and is connected to vertex $\pi_{ij}$
in tree $T_i$. Thus, if client $j$ is assigned to tree $T_i$ then the
load of $T_i$ increases by $p_{ij}$, and all edges on the path in $T_i$
from $\pi_{ij}$ to its root must be opened. The objective is to minimize
the total cost of opening the edges, subject to the capacity constraints
that the total load on tree $T_i$ is at most $u_i$.

The capacitated multicast problem generalizes the CCFL problem. Indeed,
let each machine $i\in [m]$ correspond to a two-level tree $T_i$ with
capacity $u_i$, where tree $T_i$ has a single edge $r_i$ incident to the
root, and $n$ leaves corresponding to the clients. Edge $r_i$ has
opening cost $c_i$, and the leaf edge corresponding to client $j$ has
opening cost $a_{ij}$. The load of client $j$ in tree $T_i$ is $p_{ij}$.
It is easy to check that a feasible solution to this CMC problem
instance corresponds precisely to a CCFL solution with precisely the
same cost.

In this section, we generalize the solution from the previous section to
give the following result:

\begin{thm}
  \label{thm:cmp}
  There is a randomized online algorithm that given any instance of the
  capacitated multicast problem on $d$-level trees, and a bound $C$ on
  its optimal cost, computes a solution of cost $O(\log^2m \cdot \log
  mn)\cdot C$ with congestion $O((d+\log^2m) \cdot \log mn)$.
\end{thm}

The \emph{congestion} of a solution is the maximum (over all
facilities) of the multiplicative factor by which the
capacity is violated.

The proof of this theorem will occupy the rest of this section. The main
idea is similar: we solve a convex programming relaxation of this
problem in an online fashion, and show how to round the solution online
as well. However, these will require some ideas over and above those
used in the previous section.

First, the convex relaxation. It will be convenient to augment each tree
$T_i$ as follows. For each client $j$ with $p_{ij}\le u_i$ (i.e., that
can be feasibly assigned to $T_i$), we introduce a new leaf vertex
$v_{ij}$ connected to vertex $\pi_{ij}\in T_i$ via an edge of zero cost.
These new leaf vertices $v_{ij}$ are assigned a \emph{vertex weight}
$p_{v_{ij}} := p_{ij}$, whereas all the original vertices of the trees
are given zero weight.  To minimize extra notation, we refer to these
augmented trees also as $T_i$. Finally we merge the roots of these trees
$T_i$ into a single root vertex $r$ to get a new tree $T = (V,E)$. For
client $j$, let $V_j = \{ v_{ij} \mid i \in [m] ~s.t.~ p_{ij}\le u_i\}$
denote the leaves in $T$ corresponding to client $j$.

For any edge $e\in E$, denote the subtree of $T$ below edge $e$ by
$T^e$. Observe that if $e$ was in $T_i$ then $T^e$ is a subtree of the
$i^{th}$ tree $T_i$. In this case, we use the notation $j \in T^e$ to
denote that $v_{ij} \in T^e$. For each vertex $v\in
V\setminus \{r\}$, its parent in $T$ is denoted $\tau(v)$.

Our fractional relaxation has a variable $x_e$ for each edge $e \in
E$. For brevity, we use $y_{ij} := x_{(v_{ij},\tau(v_{ij}))}$ to denote
the variable for the edge connecting the leaf-node corresponding to
client $j$ in tree $T_i$ to its parent. The $x_e$ variables naturally
denote the ``opening'' of edges, and the $y_{ij}$ variables denote the
assignment of clients to trees.  The objective is the following convex
function:
\begin{gather}
  g(x) \,\,=\,\, \left(\sum_{i=1}^m \sum_{e\in T_i} c_e \left( x_e +
      \frac{2}{u_i} \sum_{v\in T^e} p_{v} \cdot x_{v,\tau(v)}
    \right)\right)^p \, + \, \sum_{i=1}^{m}\left(\frac{2C}{u_i}\cdot
    \sum_{v\in T_i} p_{v} \cdot x_{v,\tau(v)}\right)^p.
\end{gather}
In the above expression, we choose $p=\Theta(\log m)$. Using the facts
that the weights $p_v$ are defined only for the new leaf nodes, and that
leaf edges are denoted by the $y_{ij}$ variables, we can write the above
expression equivalently as follows:
\begin{gather}
  g(x) \,\,=\,\, g(x,y) \,\,=\,\, \left(\sum_{i=1}^m \sum_{e\in T_i} c_e
    \left( x_e + \frac{2}{u_i} \sum_{j\in T^e} p_{ij} \cdot y_{ij}
    \right)\right)^p \, + \, \sum_{i=1}^{m}\left(\frac{2C}{u_i}\cdot
    \sum_{j\in T^e} p_{ij} \cdot y_{ij}\right)^p.
\end{gather}

We will solve the following convex covering program:
\begin{alignat*}{2}
  \min & \quad g(x) \\
  s.t. & \quad \textstyle \sum_{e\in \delta(S)} x_e \geq 1, \qquad
  \forall \, \, V_j\subseteq S \subseteq V\setminus \{r\},~~\forall j\in [n]\\
  &\quad x \geq 0.
\end{alignat*}
The constraints say that the min-cut between the root and the nodes in
the set $V_j$, which contains all the nodes corresponding to client $j$
in the various trees, is at least~$1$---i.e., $j$ is (fractionally)
connected at least to unit extent. Much as in Section~\ref{sec:CCFL}, we
deal with the exponential number of covering constraints as follows: we
relax the covering constraints to $\frac12$ (instead of one). Upon
arrival of client $j$, we add covering constraints based on the
following procedure.

\begin{quote}
  While there is some $V_j \subseteq S\subseteq V(G)\setminus \{r\}$
  with $\big(\sum_{e\in \delta(S)} x_e < \frac12\big)$, do: \\
  $~~~~~~$ Add constraint $\sum_{e\in \delta(S)} x_e \ge 1$, and update
  $(x,y)$ according to Theorem~\ref{thm:p-norm}.
\end{quote}

Note that given a current solution $x$, one can find such a ``violated
constraint'' (if there is one) by a minimum-cut subroutine, which takes
polynomial time. The number of iterations of the above procedure is at
most $2|E|$, because $\sum_{e\in E} \min\{x_e,1\}$ increases by at least
$\frac12$ in each iteration, but it starts at $0$ and stays at most
$|E|$. Moreover, twice the solution is always a feasible solution, which
implies an $O(\log^2 m)$-competitive online algorithm for the fractional
problem.

\subsubsection{Rounding the Solution Online}
\label{sec:cmc-round}

For the online rounding, define some modified variables. For each client
$j\in[n]$, compute a unit-flow ${\cal F}_j$ from the set $V_j$ to the
root $r$ in the tree $T$ with edge-capacities $2x$; note that the
fractional solution guarantees this flow exists. Let $f_{e}^j$ be the
amount of flow on edge $e$ in ${\cal F}_j$, and define $f_e :=\max_{j}
f^j_e$. Note that the $f_e$ values are monotone non-decreasing as we go
up the tree $T$. Now set:
\begin{gather}
  \ox_e = \max\bigg\{ f_{e} ~~,~~ \frac{2}{u_i}\sum_{j\in T^e} p_{ij}
  \cdot y_{ij} \bigg\},\quad \forall e\in T_i,~ \forall i\in[m].
\end{gather}
Also define $\oy_{ij} = \ox_{(v_{ij},\tau(v_{ij}))}$ for any client $j$ and
tree $T_i$, to capture the assignment of clients of trees.

\begin{claim}
  The variables $\ox_e$ are monotone non-decreasing up the tree $T$.
\end{claim}
\begin{proof}
  The flow values $f_e$ are monotone non-decreasing up the tree. Also,
  for any tree $T_i$ and any edge $e \in T_i$, the quantity $f'_e :=
  \frac{2}{u_i}\sum_{j\in T^e} p_{ij} \cdot y_{ij}$ is also monotone
  non-decreasing up the tree, since it is the sum of non-negative
  quantities over larger subtrees. Since tree $T$ is obtained by merging
  the trees $T_i$ at the root, the monotonicity of
  $\ox_e=\max\{f_e,f'_e\}$ maintained.
\end{proof}

Note that $(\ox,\oy)$ clearly satisfies:
\begin{alignat}{2}
  \oy_{ij} &\leq \ox_i & \qquad\qquad
  & \forall j\in T^e, ~ \forall
  e\in T_i, ~ \forall i\in [m].
  \label{eq:cmp:0} \\
  \sum_{j\in T^e}\, p_{ij}\cdot \oy_{ij}  = \sum_{v\in
    T^e}\, p_{v}\cdot \ox_{v,\tau(v)}  &\le  u_i \cdot \ox_e
  & & \forall e\in T_i,\,\, \forall i\in [m].
  \label{eq:cmp:1} \\
  \ox_e &\le \ox_{\tau(e)} & & \forall \, e \in T.
  \label{eq:cmp:2} \\
  \sum_{i=1}^m \oy_{ij} &\ge 1 && \forall j\in [n].
  \label{eq:cmp:3}
\end{alignat}

\begin{claim}
  \label{cl:cmp-LB}
  Assuming there exists an integral solution to the CMC problem instance
  having cost at most $C$, the following inequalities hold with
  $\alpha=O(\log^2m)$:
  \begin{alignat}{2}
    \sum_{e\in T} c_e\cdot \ox_e &\le  4\alpha \cdot C
    && \label{eq:cmp:4} \\
    \sum_{j\in T_i}\, p_{ij}\cdot \oy_{ij} = \sum_{v\in
      T_i}\, p_{v}\cdot \ox_{v,\tau(v)} &\le 4\alpha \cdot
    u_i, & \qquad\qquad &\forall i\in[m]. \label{eq:cmp:5}
  \end{alignat}
\end{claim}

\begin{proof}
  The optimal integral solution of the current CMC problem instance has
  cost most $C$, hence the optimal fractional value of our convex
  covering problem is at most $(3C)^p+ m\cdot (2C)^p\le (m+1)(3C)^p$, and
  our $\alpha$-competitive algorithm ensures that $g(x,y)\le
  \alpha^p\cdot (m+1)(3C)^p \leq (4\alpha C)^p$ for $p \geq \log_{4/3}
  (m+1) = \Theta(\log m)$. This, in turn, implies that
  \begin{gather}
    \sum_{e\in T} c_e\cdot \ox_e \le \sum_{e\in T} c_e \bigg(x_e +
      \frac{2}{u_i}\sum_{j\in T^e} p_{ij} \cdot y_{ij} \bigg) \le
    g(x,y)^{1/p} \\
    \frac{C}{u_i}\sum_{j\in T_i} p_{ij} \cdot \oy_{ij} \le \frac{2C}{u_i}
    \sum_{j\in T_i} p_{ij} \cdot y_{ij} \le g(x,y)^{1/p}
  \end{gather}
  which proves the claim.
\end{proof}

Having defined these convenient modified variables, the rounding
proceeds as follows. For each tree $T_i$, the edges $F_i = \{e\in T_i
\mid \ox_e\ge 1\}$ form a rooted subtree, by the monotonicity of the
$\ox$ values. We include the edges in $F_i$ in the solution
deterministically. For the rest of the edges, we perform the following
experiment $\beta := \Theta(d\cdot \log mn)$ times independently, and
take the union of the edges picked.
\begin{quote}
  For each tree $T_i$, independently:
  \begin{enumerate}
  \item[(i)] For each edge $e\in T_i \setminus F_i$, pick it independently
    with probability $\frac{\ox_e}{\ox_{\tau(e)}}$, where we use
    $\tau(e)$ to denote the parent edge of $e$. An edge $e$ whose parent
    edge does not lie in $T_i\setminus F_i$ is chosen with probability
    $\ox_e$.
  \item[(ii)] If the load for tree $T_i$ exceeds $8(d+4\alpha)\cdot u_i$,
    declare failure for all clients assigned to $T_i$.
  \end{enumerate}
\end{quote}
The rounding in step~(i) is from Garg et al.~\cite{GKR98} and hence is
often called the GKR-rounding; it can be implemented online using ideas
from~\cite{AAABN-talg06}. Note that there is some probability that for
some client $j$ we may declare failure for all $\beta$ experiments. In
that case we can choose the path in that tree $T_i$ for client $j$ which
is cheapest subject to $p_{ij} \leq u_i$.

A client $j$ is \emph{assigned in tree $T_i$} if all edges on the path
from $v_{ij}$ to $r$ are picked in $T_i$ during Step~1, and if we don't
declare failure in Step~2; a client is \emph{assigned} if it is assigned
in at least one tree.  We first show that there is a good probability of
any client being assigned in one run of the random experiment above.
\begin{claim}
  \label{cl:cmp-alteration}
  For any client $j$, $\Pr[j\mbox{ assigned in one run}] \ge \frac12$.
\end{claim}

\begin{proof}
  It is easy to check that for any tree $T_i$, $\Pr[j\text{ assigned to
    $T_i$ in Step~1}]=\min\{\oy_{ij},1\}$.  Since the random choices in
  different trees $T_i$ are independent,
  \begin{gather*}
    \Pr[j\mbox{ not assigned to any tree in step 1}] = \prod_{i=1}^m
    \left( 1- \min\{\oy_{ij},1\} \right) \le \mathrm{e}^{-\sum_{i=1}^m \oy_{ij}}
    \le_{\eqref{eq:cmp:3}} \frac1{\mathrm{e}}.
  \end{gather*}
  Next, we claim that conditioned on $j$ being assigned in tree $T_i$ in
  Step~1 (i.e., on all edges on the path $P_{ij}$ from the root of $T_i$
  to $v_{ij}$ being chosen in the solution), the conditional probability
  it is rejected in Step~2 is at most $1/8$, i.e.,
  \begin{equation}\label{eq:cmp-alteration}
    \Pr[j\mbox{ rejected in step 2 } | \, j \mbox{ assigned to $T_i$ in step 1}]\quad \le \quad \frac18.
  \end{equation}
  This would imply that $j$ is assigned in at least one tree with
  probability $(1-1/\mathrm{e})$, and survives rejection in that tree
  with probability $7/8$, giving $(1-\frac{1}{\mathrm{e}})\frac78 \ge
  \frac12$.

  To prove~\eqref{eq:cmp-alteration}, let edges $e_1,\cdots, e_k$ be the
  edges of $T_i\setminus F_i$ on path $P_{ij}$ at increasing distance
  from the root; hence $e_k = (\tau(v_{ij}), v_{ij})$. For
  $h=1,\cdots,k$, define subtree $S_h:=T^{e_h}\setminus T^{e_{h+1}}$
  which consists of all nodes whose path to the root first intersects
  with $P_{ij}$ at the edge $e_h$. By the properties of the GKR
  rounding, we have in Step~1:
  \begin{gather*}
    E[\mbox{load from }S_h \mid  j\mbox{ assigned to $T_i$ in step~1} ]  =
    \sum_{\ell\in S_h} p_{i\ell}\cdot
    \frac{\oy_{i\ell}}{\ox_{e_h}}\quad \le \quad \frac1{\ox_{e_h}} \cdot
    \sum_{\ell\in T^{e_h}} p_{i\ell}\cdot \oy_{i\ell} \quad
    \le_{\eqref{eq:cmp:1}} \quad u_i.
  \end{gather*}
  Summing the expression above for all $h=1,\cdots,k$,
  \begin{gather*}
    E[\mbox{load from }\cup_{h=1}^k S_h \mid j\mbox{ assigned to
    $T_i$ in step~1}] \quad \le \quad k\cdot u_i \quad \le \quad d\cdot u_i,
  \end{gather*}
  since the tree has depth at most $d$.

  This bounds the expected load of those clients whose paths to the root
  share an edge with $P_{ij}$ in tree $T_i$. For any other client
  $\ell$, the conditioning does not matter, and hence
  \begin{gather*}
    \Pr[\ell\mbox{ assigned to }T_i\mid j\mbox{ assigned to
    $T_i$ in step~1} ]=\Pr[\ell\mbox{ assigned to }T_i] = \oy_{i\ell}.
  \end{gather*}
  So using~\eqref{eq:cmp:5},
  \begin{gather*}
    E[\mbox{load from }[n]\setminus j\setminus \cup_{h=1}^k S_h \, | \,
    j\mbox{ assigned to $T_i$ in step~1}] \quad \le \quad \sum_{\ell\in T_i}
    p_{i\ell}\cdot \oy_{i\ell} \quad \le \quad 4\alpha\cdot u_i
  \end{gather*}
  Thus the total expected load from $[n]\setminus j$ conditioned on $j$
  being assigned to $T_i$ in Step~1 is at most $(d+4\alpha)\cdot
  u_i$. Markov's inequality now implies~\eqref{eq:cmp-alteration}, and
  hence the claim.
\end{proof}

By Step~2 of the algorithm, we immediately have:
\begin{claim}
  \label{cl:cmp-load}
  The load assigned to tree $T_i$ is at most $8\beta(d+4\alpha)\cdot
  u_i$, for each $i\in[m]$.
\end{claim}

\begin{claim}
  \label{cl:cmp-open}
  The expected opening cost is at most $4\alpha\beta \cdot C$.
\end{claim}
\begin{proof}
  The expected cost of edges chosen in each of the $\beta$ independent
  trials is at most $\sum_e c_e\cdot \ox_e \le 4\alpha C$
  using~\eqref{eq:cmp:4}. Summing the cost over all trials gives the
  claim.
\end{proof}

\begin{claim}
  For any client $j$, the probability that $j$ is unassigned is at most
  $\frac{1}{mn^2}$.
\end{claim}
\begin{proof}
  By Claim~\ref{cl:cmp-alteration}, the probability of $j$ being
  unassigned in one trial is at most $\frac12$. Since there are
  $\beta=\Theta(\log mn)$ independent trials, the claim follows.
\end{proof}

\begin{proofof}{Theorem~\ref{thm:cmp}}
  By Claims~\ref{cl:cmp-load} and~\ref{cl:cmp-open}, we know that the
  cost and load of the solution is at most the claimed bounds. Moreover,
  we know that the probability of the client not being assigned to any
  of the trees is at most $\frac{1}{mn^2}$. Since this will increase the
 load of some tree $i$ by at most $u_i$ and the cost by at most $OPT$,
  and happens with probability at most $\frac{1}{mn^2}$, this increases
  the expected cost and congestion by a negligible factor.
\end{proofof}


\subsection{Set Cover with Set Requests}
\label{sec:sc-set-requests}

We consider here the {\em online set cover with set requests} (SCSR)
problem first consideed by Bhawalkar et al.~\cite{BGP14}, which is
defined as follows. We are given a universe $U$ of $n$ \emph{resources},
and a collection of $m$ \emph{facilities}, where each facility $i\in[m]$
is specified by (i) a subset $S_i\sse U$ of resources (ii) opening cost
$c_i$ and (iii) capacity $u_i$. The resources and facilities are given
up-front. Now, a sequence of $k$ {\em requests} arrive over time. Each
request $j\in[k]$ requires some subset $R_j\sse U$ of resources. The
request has to be served by assigning it to some collection $F_j \sse
[m]$ of facilities whose sets collectively cover $R_j$, i.e., $R_j\sse
\cup_{i\in F_j} S_i$. Note that these facilities have to be open, and we
incur the cost of these facilities. Moreover, if a facility $i$ is used
to serve client $j$, this contributes to the load of facility $i$, and
this total load must be at most the capacity $u_i$.

As in previous sections, we give an algorithm to compute a solution
online which violates the capacity constraint by some factor. Our main result for this
problem is the following:
\begin{thm}\label{thm:scsr}
  There is a randomized online algorithm that given any instance of the
  \emph{set cover with set requests} problem and a bound $C$ on its
  optimal cost, computes a solution of cost $O(\log^2m \cdot \log
  mnk)\cdot C$ with congestion $O(\log^2m \log mnk)$.
\end{thm}

The ideas --- for both the convex relaxation and the rounding --- are
very similar to that for CCFL; hence we only sketch the main ideas
here. For the fractional relaxation, there is a variable $x_i$ for each
facility $i\in[m]$ denoting if the facility is opened. For each request
$j\in[k]$ and facility $i$ there is a variable $y_{ij}$ that denotes if
request $j$ is connected to facility $i$. We set $p=\Theta(\log m)$ and
the objective is:
\begin{gather*}
  g(x,y) \,\,=\,\, \left(\sum_{i}^{}c_i \left(x_i + \frac{\sum_j
        y_{ij}}{u_i}\right)\right)^p \,+\, C^p \cdot \sum_{i}^{} \left(
    x_i + \frac{1}{u_i} \sum_{j}^{ }y_{ij}\right)^p.
\end{gather*}
We define the following convex covering program, where we use
$F(\ell):=\{i\in[m] \mid \ell\in S_i\}$ for each resource $\ell\in U$.
\begin{alignat*}{2}
  \min & \quad g(x,y)\\
  s.t. & \quad \sum_{i\in T} x_i + \sum_{i\in F(\ell)\setminus T}y_{ij}
  \geq 1, \qquad \forall T\sse F(\ell),\,\, \forall \ell\in R_j,\,\,
  \forall j\in [k],\\
  &\quad y, x \geq 0.
\end{alignat*}
We can solve this convex program in an online fashion, much as in
Theorem~\ref{thm:p-norm}. Now for the rounding: we maintain the
following modified variables:
\begin{gather*}
  \overline{y}_{ij}=\min\{y_{ij},x_i\}, \quad \forall i,j. \\
\ox_i = \max\left\{ x_i, \frac{\sum_j y_{ij}}{u_i} \right\},\quad \forall i.
\end{gather*}
Note that $(\ox,\oy)$ clearly satisfy:
\begin{alignat}{2}
  \sum_j \oy_{ij} &\le u_i\cdot \ox_i & \qquad\qquad & \forall i.
  \label{eq:scsr:1} \\
  \sum_{i\in F(\ell)} \oy_{ij} &\ge \frac12 && \forall \ell\in R_j,\,\,
  \forall j\in [k].
  \label{eq:scsr:2} \\
  \oy_{ij} &\le \ox_i && \forall i, j.
  \label{eq:scsr:3}
\end{alignat}

\begin{claim}\label{cl:scsr-LB}
  Assuming that there is an integral solution to the SCSR instance
  having cost at most $C$, the following inequalities hold for
  $\alpha=O(\log^2m)$:
  \begin{gather}
    \sum_i c_i\cdot \ox_i \le 4\alpha \cdot C.
    \label{eq:scsr:4}\\
    \ox_{i} \le 4\alpha,\quad \forall i.
    \label{eq:scsr:6}
  \end{gather}
\end{claim}
The proof is similar to Claim~\ref{cl:ccfl-LB} and omitted.

The final randomized rounding is the same as for CCFL. For each facility
$i$, set $X_i$ to one with probability $\min\{4\log (mnk) \cdot \ox_i,
1\}$. Let $\oF$ denote the set of {\em fixed} facilities,
i.e. $\Pr[X_i=1]=1$. So $\oF=\{i : \ox_{i} > \frac{1}{4\log
  (mnk)}\}$. For each request $i$ and facility $j$, set $Z_{ij}$ to one
with probability:
\begin{gather*}
  Pr[ Z_{ij}=1] \,\,=\,\, \left\{
    \begin{array}{ll}
      \min\{4\log (mnk)\cdot \oy_{ij}, 1\} & \mbox{ if }i\in \oF, \\
      \frac{\oy_{ij}}{\ox_i} & \mbox{ otherwise.}
    \end{array}
  \right.
\end{gather*}

All the above random variables are independent. Each request $i$ gets
connected to all facilities $j$ with $X_i\cdot Z_{ij}=1$.  The analysis
of the rounding is also identical to that of CCFL, and is omitted. This
completes the proof of Theorem~\ref{thm:scsr}.


\section{Profit maximization with non-separable production costs}
\label{sec:profit-max}

In this section we consider a profit maximization problem (called PMPC) for a single
seller with production costs for items. There are $m$ items that the
seller can produce and sell. The production levels are given by a vector
$\mu \in \RR_+^m$; the total cost incurred by the seller to produce $\mu_j$
units of every item $j\in [m]$ is $g(\mu)$ for some \emph{production cost
  function} $g:\RR_+^m\rightarrow \RR_+$. In this work we allow for
functions $g$ which are convex and monotone in a certain
sense\footnote{The formal conditions on $g$ appear in
  Assumption~\ref{asm:g-mono}.}. There are $n$ buyers who arrive
online. Each buyer $i\in[n]$ is interested in certain subsets of items
(a.k.a.\ \emph{bundles}) which belong to some set family $\cS_i \sse
2^{[m]}$. The extent of interest of buyer $i$ for subset $S \in \cS_i$
is given by $v_i(S)$, where $v_i:\cS_i\rightarrow \RR_+$ is her
\emph{valuation function}.

If buyer $i$ is allocated a subset $T\in \cS_i$ of items, he pays the
seller his valuation $v_i(T)$. Consider the optimization problem for the
seller: he must produce some items and allocate bundles to buyers so as
to maximize the profit $\sum_{i=1}^n v_i(T_i) - g(\mu)$, where $T_i\in
\cS_i$ denotes the bundle allocated to buyer $i$ and $\mu= \sum_{i=1}^n
\chi_{T_i}\in \RR^m$ is the total quantity of all items produced. (Here
$\chi_S \in \{0,1\}^m$ is the characteristic function of the set $S$.)
Observe that in this paper we consider a non-strategic setting, where
the valuation of each buyer is known to the seller; this differs from an
auction setting, where the seller has to allocate items to buyers
without knowledge of the true valuation, and the buyers may have an
incentive to mis-report their true valuations.

This class of maximization problems with production costs was introduced
by Blum et al.~\cite{BGMS11} and more recently studied by Huang and Kim
\cite{HK15}. Both these works dealt with the online \emph{auction}
setting, but in both works they considered a special case where the
production costs were separable over items; i.e, where $g(\mu) = \sum_j
g_j(\mu_j)$ for some convex functions $g_j(\cdot)$. In contrast, we can
handle general production costs $g(\cdot)$, but we do not consider the
auction setting. Our main result is for the \emph{fractional version} of
the problem where the allocation to each buyer $i$ is allowed to be any
point in the convex hull of the $\cS_i$. In particular, we want to solve
following convex program in an online fashion:
\begin{alignat}{2}
  \mbox{maximize } \sum_{i=1}^{n} \sum_{T\in \cS_i} &v_i(T)\cdot y_{iT} & \,\,-\,\, g(\mu)   &
  \tag{$D$} \label{lp:prod-D}\\
  \sum_{T\in \cS_i} y_{iT} & \le 1 & & \forall \,
  i \in [n],   \label{lp:prod-D1} \\
\sum_{i=1}^n \sum_{T\in \cS_i} \mathbf{1}_{j\in T}\cdot y_{iT} - \mu_j & \le 0& \quad & \forall j\in [m],  \label{lp:prod-D2}  \\
  y,\mu &\geq 0. & \qquad &
\end{alignat}

\def\ou{u}

Note that this problem looks like the dual of the covering problems we
have been studying in previous sections, and hence is suggestively
called~\eqref{lp:prod-D}. Consider the following ``dual'' program that
gives an upper bound on the value of~(\ref{lp:prod-D}).
\begin{alignat}{2}
  \mbox{minimize } \sum_{i=1}^{n} \ou_i \,\,+\,\, & g^\star(x) &  &
  \tag{$P$} \label{lp:prod-P}\\
  \ou_i + \sum_{j\in T} x_j & \ge v_i(T) & & \forall \,
  i \in [n], \,\,\forall T\in \cS_i,  \label{lp:prod-P1} \\
  \ou,x &\geq 0. & \qquad &
\end{alignat}
Again, to be consistent with our general framework, we refer to this
minimization (covering) problem as the ``primal''~\eqref{lp:prod-P}.


Notice that this primal-dual pair falls into the general framework of
Section~\ref{sec:general} if we set
\begin{gather*}
  f(\ou,x)\quad :=\quad \sum_{i=1}^{n} \ou_i \,+\, g^\star(x).
\end{gather*}
Indeed, if we were to construct the Fenchel dual of~\eqref{lp:prod-P} as in
Section~\ref{sec:general}, we would again arrive at~\eqref{lp:prod-D}
after some simplification (using the fact that $g^{**}=g$ for any convex
function $g$ with subgradients\footnote{A \emph{subgradient} of
  $g:\RR^m\rightarrow \RR$ at $u$ is a vector $V_u\in \RR^m$ such that
  $g(w)\ge g(u)+V_u^T (w-u)$ for all $w\in \RR^m$.}~\cite{Rock}). In
order to apply now our framework, we assume that $f$ is continuous,
differentiable and satisfies $\nabla f(z) \ge \nabla f(z')$ for all
$z\ge z'$. This translates to the following assumptions on the
production function $g$:
\begin{assumption}
  \label{asm:g-mono}
  Function $g^\star:\RR^m_+\rightarrow \RR_+$ (recall $g^\star(x)=\sup_\mu
  \{x^T\mu-g(\mu)\}$) is monotone, convex, continuous, differentiable
  and has $\nabla g^\star(x) \ge \nabla g^\star(x')$ for all $x\ge x'$.
\end{assumption}

Since we require irrevocable allocations, we cannot use the primal-dual
algorithm from Section~\ref{sec:alg-description}, since that algorithm could
decrease the dual variables $y_{iT}$. Instead, we use the algorithm from
Section~\ref{sec:monotone} which ensures both primal and dual
variables are monotonically raised. We can now use the competitive ratio
from~(\ref{eq:mono-cr})---when $g^\star(0)=0$ this ratio is at least
\begin{equation}
\max_{c > 0}\left\{ \min_{z} \left(\frac{\min_{\ell=1}^{n}\left\{\frac{\nabla_\ell g^\star(z)}{\nabla_\ell g^\star(cz)}\right\}}{2\ln(1+\rho d)}\right) - \max_{z}
\left(\frac{ z^\trans \nabla g^\star(z) - g^\star(z)
  }{g^\star(c z)}\right) \right\} \label{eq:prod-comp-ratio}
\end{equation}
In this expression, recall that $d$ is the row-sparsity of the covering
constraints in~\eqref{lp:prod-P}, i.e. $d=1+\max_{T\in \cup \cS_i} |T|$. And the term $\rho$ is the ratio
between the maximum and minimum (non-zero) valuations any player $i$ has
for any set in $\cS_i$. In other words,
\begin{gather}
  \rho \quad \le \quad R\,:=\, \,\, \frac{ \max \left\{ v_i(T) : T\in
      \cS_i,\, i\in [n]\right\}}{\min \left\{v_i(T) : T\in \cS_i,
      v_i(T)>0,\, i\in [n]\right\}}. \label{eq:def-R}
\end{gather}

\subsection{An Efficient Algorithm for~\eqref{lp:prod-D}}

To solve the primal-dual convex programs using our general framework in
polynomial time, we need access to the following oracle:
\begin{quote}
  \emph{Oracle:} Given vectors $(\ou, x)$, and an index $i$, find a
  set $T\in \cS_i$ such that
  \begin{gather}
    \textstyle \ou_i + \left(\sum_{j\in T} x_j -
      v_i(T)\right) < 0, \label{eq:sep-oracle-prod}
  \end{gather}
 or else report that no such set exists.
\end{quote}
Given such an oracle, we maintain a $(\ou,x)$ such that $(2\ou, 2x)$ is
feasible for~(\ref{lp:prod-P}) as follows. When a new buyer $i$ arrives,
we use the oracle on $(2\ou, 2x)$. While it returns a set $T \in \cS_i$,
we update $(\ou, x)$ to satisfy the constraint~\eqref{lp:prod-P1}. Else
we know that $(2\ou, 2x)$ is a feasible solution for~(\ref{lp:prod-P}).
This scaling by a factor of~$2$ allows us to bound the number of
iterations as follows: when buyer $i$ arrives, define $Q_i =
\min\{\ou_i, V^i_{max}\}+ \sum_{j=1}^m \min\{x_j, V^i_{max}\}$ where
$V^i_{max}=\max \left\{ v_i(T) \mid T\in \cS_i\right\}$. Note that
$Q_i\le (m+1)V^i_{max}$ 
and $Q_i$ increases by at least $V^i_{min}/2$ in each
iteration where $V^i_{min}=\min \left\{v_i(T) \mid T\in \cS_i,
  v_i(T)>0\right\}$. So the number of iterations is at most $O(mR)$
where $R$ is defined in~(\ref{eq:def-R}). This gives us a
polynomial-time online algorithm if $R$ is polynomially bounded.

What properties do we need from the collection $\cS_i$ and valuation
functions $v_i$ such that can we implement the oracle efficiently? Here
are some cases when this is possible.
\begin{itemize}
\item \emph{Small $\cS_i$.} If each $|\cS_i|$ is polynomially bounded then we can
  solve~\eqref{eq:sep-oracle-prod} just by enumeration. An example is when each buyer is ``single-minded'' i.e. she wants exactly one bundle.

\item {\em Supermodular valuations.} Here, buyer $i$ has $\cS_i=2^{[m]}$
  and $v_i:2^{[m]}\rightarrow \RR_+$ is supermodular,
  i.e. $v_i(T_1)+v_i(T_2) \le v_i(T_1\cup T_2)+v_i(T_1\cap T_2)$ for all
  $T_1,T_2\sse [m]$. In this case, we can
  solve~\eqref{eq:sep-oracle-prod} using polynomial-time algorithms for
  submodular minimization~\cite{Schrijver-book}, since the expression inside the
  minimum is a linear function minus a supermodular function.

\item {\em Matroid constrained valuations.}  In this setting, each buyer
  $i$ has some value $v_{ij}$ for each item $j\in [m]$ and the feasible
  bundles $\cS_i$ are independent sets of some matroid.\footnote{An
    alternative description of such valuation functions is to have
    $\cS'_i=2^{[m]}$ and $v'_i(T) = $ maximum weight independent subset
    of $T$ (where each item $j$ has weight $v_{ij}$). Viewed this way,
    the buyer's valuation is a weighted matroid rank function which is a
    special submodular function.}  Here we can
  solve~\eqref{eq:sep-oracle-prod} by maximizing a linear function over
  a matroid. This is because the minimization
  \begin{gather*}
    \min_{T\in \cS_i} \, \bigg(\sum_{j\in T} x_j - v_i(T)\bigg) \quad =
    \quad \min_{T\in \cS_i} \, \sum_{j\in T} (x_j - v_{ij}) \quad =
    \quad -\max_{T\in \cS_i} (v_{ij}-x_j),
  \end{gather*}
  can be done in polynomial time~\cite{Schrijver-book}.
\end{itemize}


\subsection{Online Rounding}

We now have a deterministic online algorithm for~\eqref{lp:prod-D} with competitive ratio as given in~\eqref{eq:prod-comp-ratio}. Moreover, this algorithm runs in polynomial time for many special cases. Here we show how the fractional online solution can be rounded to give integral allocations. We make the following additional assumption on the production costs.

\begin{assumption}   \label{asm:g-rnd}
There is a constant $\beta>1$ such that $g(a \mu)\le a^{\beta}\cdot g(\mu)$ $\forall\, 0<a<1,\, \mu\in \RR^m_+$.
\end{assumption}

\begin{thm}\label{thm:prod-integral}
For any $\epsilon\in (0,1)$ there is a randomized online algorithm for PMPC under Assumptions~\ref{asm:g-mono} and~\ref{asm:g-rnd}   that achieves expected profit at least $(1+\epsilon)^{-2-\frac{2}{\beta-1}}\cdot \frac{\OPT}{\alpha} - \frac{\OPT}{m\alpha} - g(L\cdot \mathbf{1})$, where $\OPT$ is the offline optimal profit, $\alpha$ is the fractional competitive ratio and $L=O(\frac{\log m}{\epsilon^2})$.
\end{thm}

Note that the additive error term $g(L\cdot \mathbf{1})$ is independent of the number $n$ of buyers: it depends only on the number $m$ of  items and the production function $g$. We also give an example below which shows that any rounding algorithm for~\eqref{lp:prod-D} must incur some such additive error.

We now describe the rounding algorithm. Let $\epsilon\in(0,1)$ be any value; set $a=(1+\epsilon)^{-2-\frac{2}{\beta-1}}$. The rounding algorithm scales the fractional allocation $y$ by factor $a<1$ and performs randomized rounding. Let $M\in \mathbb{Z}_+^m$ denote the (integral) quantities of different items produced at any point in the online rounding.  Upon arrival of buyer $i$, the algorithm does the following.
\begin{enumerate}
\item Update fractional solution $(y,\mu)$ according to the fractional online algorithm.
\item If $M_j > (1+\epsilon) a\mu_j + \frac6\epsilon \log m$ for any $j\in[m]$ then skip.
\item Else, allocate set $T \in \cS_i$ to buyer $i$ with probability $a\cdot y_{iT}$.
\end{enumerate}

\begin{claim}\label{clm:prod-rnd}
$\Pr[M_j>(1+\epsilon) a\mu_j + \frac6\epsilon \log m] \le \frac{1}{m^2}$ for all items $j\in[m]$ and $\epsilon\in (0,1)$.
\end{claim}
\begin{proof}
Fix $j\in[m]$ and $\epsilon\in (0,1)$. Note that $M_j$ is the sum of independent $0-1$ random variables with $\E[M_j]\le a\cdot \mu_j$. The claim now follows by Chernoff bound.
\end{proof}

Below $\ell:=\frac6\epsilon \log m +1$ and $L:=(1+\frac1\epsilon)\cdot \ell=O(\frac{\log m}{\epsilon^2})$.

\begin{lemma}\label{lem:prod-rnd} The expected objective of the integral allocation is
at least
$$a(1-\frac1m)\cdot \sum_{i=1}^{n} \sum_{T\in \cS_i} v_i(T)\cdot y_{iT} - a \cdot g(\mu) - g(L\cdot \mathbf{1}).$$
\end{lemma}
\begin{proof} Note that the algorithm ensures (in step~2 above) that $M\le (1+\epsilon)a\cdot \mu + \ell\cdot \mathbf{1}$. So with probability one, the production cost  is at most:
\begin{eqnarray*}
g((1+\epsilon)a\cdot \mu + \ell\cdot \mathbf{1}) & = & g\left(\frac{1}{1+\epsilon}\cdot (1+\epsilon)^2a\cdot \mu + \frac{\epsilon}{1+\epsilon}\cdot (1+\frac1\epsilon) \ell\cdot \mathbf{1}\right)\\
&\le &  g\left((1+\epsilon)^2a\cdot \mu\right) +  g\left( (1+\frac1\epsilon) \ell\cdot \mathbf{1}\right)\\
&\le & ((1+\epsilon)^2a)^\beta\cdot g(\mu) + g\left( L\cdot \mathbf{1}\right)\quad = \quad a \cdot g(\mu) + g\left( L\cdot \mathbf{1}\right).
\end{eqnarray*}
The first inequality is by convexity of $g$; the second inequality uses Assumption~\ref{asm:g-rnd} and $(1+\epsilon)^2a<1$; the last equality is by definition of $a$.

By Claim~\ref{clm:prod-rnd}, the probability that we skip some buyer $i$ is at most $\frac{1}{m}$. Thus the expected total value is at least $(1-\frac{1}{m})a \cdot \sum_{i=1}^{n} \sum_{T\in \cS_i} v_i(T) y_{iT}$. Subtracting the upper bound on the cost from the expected value, we obtain the lemma. \end{proof}

This completes the proof of Theorem~\ref{thm:prod-integral}.

\paragraph{Integrality Gap.} We note that the additive error term is necessary for any algorithm based on the convex relaxation~\eqref{lp:prod-D}. Consider a single buyer with $\cS_1=2^{[m]}$ and $v_1(T)=|T|$. Let $g(\mu) = \sum_{j=1}^m \mu_j^2$. The optimal integral allocation clearly has profit zero. However the fractional optimum is $\Omega(m)$ due to the feasible solution with $y_{1T}=2^{-m}$ for all $T\sse [m]$ and $\mu_j=\frac12$ for all $j\in[m]$. Thus any algorithm using this relaxation incurs an additive error depending on $m$.

\subsection{Examples of Production Costs}
Here we give two examples of production costs (satisfying Assumptions~\ref{asm:g-mono} and \ref{asm:g-rnd}) to which our results apply. In each case, we first show the competitive ratio obtained for the fractional convex program, and then use the rounding algorithm to obtain an integral solution.

\paragraph{Example 1.} Consider a seller who can produce items in $K$ different factories, where the $k$'th factory produces in one hour of work $p_{kj}$ units of item $j$. The production cost is the sum of $q^{th}$ powers of the work hours of the $K$ factories (specifically, we get a linear production cost for $q=1$ and the $q^{th}$ power of makespan when $q\geq\log K$).
This corresponds to the following function:
\begin{equation} \label{eq:prod-cost-eg1} g(\mu) \quad =\quad \min \left\{ \frac{1}{q}\sum_{k=1}^K z_k^q\,\, :\,\, \sum_{k=1}^K p_{kj}\cdot z_k \ge \mu_j,\, \forall j\in[m],\,\, z\ge 0\right\}.
\end{equation}
We scale the objective by $1/q$ to get a more convenient form. The dual function is:
$$g^\star(x) \quad =\quad \frac{1}{p}\sum_{k=1}^K \left(\sum_{j=1}^m p_{kj}\cdot x_j\right)^p, \qquad \mbox{where }\frac{1}{p}+\frac1q=1.$$
Applying our framework (as Assumption~\ref{asm:g-mono} is satisfied), as in Section~\ref{sec:Lp-norm-packing}, we obtain an $\alpha= O(p\log\rho d)^p$-competitive fractional online algorithm, where $\rho =R$ 
the maximum-to-minimum ratio of valuations and row-sparsity $d\le m+1$. Combined with Theorem~\ref{thm:prod-integral} (note that Assumption~\ref{asm:g-rnd} is satisfied with $\beta=q$), setting $\epsilon=\frac12$, we obtain:
\begin{cor}
There is a randomized online algorithm for PMPC with cost function~\eqref{eq:prod-cost-eg1} for $q>1$ that achieves expected profit at least $(1-\frac1m) \frac{\OPT}{O(p\log R d)^p} - g(O(\log m)\cdot \mathbf{1})$.
\end{cor}
Note that $g(O(\log m)\cdot \mathbf{1})\le K\cdot O\left(\frac{\log m}{p_{min}}\right)^q$ where $p_{min}>0$ is the minimum positive entry in $p_{kj}s$.

\paragraph{Example 2.} This deals with the dual of the above production cost. Suppose there are $K$ different  linear  cost functions: for $k\in [K]$ the $k^{th}$ cost function is given by $(c_{k1},\cdots,c_{km})$ where $c_{kj}$ is the cost per unit of item $j\in[m]$. The production cost $g$ is defined to be the (scaled) sum of $p^{th}$ powers of these $K$ different costs:
\begin{equation}\label{eq:prod-cost-eg2} g (\mu) \quad =\quad \frac{1}{p}\sum_{k=1}^K \left(\sum_{j=1}^m c_{kj}\cdot \mu_j\right)^p.
\end{equation}
This has dual:
$$g^\star(x) \quad =\quad \min \left\{ \frac{1}{q}\sum_{k=1}^K z_k^q\,\, :\,\, \sum_{k=1}^K c_{kj}\cdot z_k \ge x_j,\, \forall j\in[m],\,\, z\ge 0\right\}, \qquad \mbox{where }\frac{1}{p}+\frac1q=1.$$

The primal program~\eqref{lp:prod-P} after eliminating variables $\{x_j\}_{j=1}^m$ is given below with its dual:

\begin{minipage}[t]{0.5\textwidth}
{\small \begin{alignat}{2}
  \mbox{minimize } \sum_{i=1}^{n} \ou_i \,\,+\,\, & \frac1q \sum_{k=1}^K z_k^q  \qquad \qquad \qquad \mathbf{(P')}
  \notag & \\ 
  \ou_i + \sum_{j\in T} \sum_{k=1}^K c_{kj}\cdot z_k & \ge v_i(T),\quad  \forall \,
  i \in [n], \,\,\forall T\in \cS_i,  & \notag \\
  \ou,z &\geq 0. & \notag
\end{alignat}}
\end{minipage}
\begin{minipage}[t]{0.5\textwidth}
{\small \begin{alignat}{2}
  \mbox{maximize } \sum_{i=1}^{n} \sum_{T\in \cS_i} &v_i(T)\cdot y_{iT}  \,\,-\,\, \frac1p \sum_{k=1}^K \lambda_k^p   \qquad \mathbf{(D')} \notag &\\
  \sum_{T\in \cS_i} y_{iT}  \le 1,   & \quad \forall \, i \in [n],   \notag \\
\sum_{i=1}^n \sum_{T\in \cS_i} (\sum_{j\in T} c_{kj} ) & \cdot y_{iT} - \lambda_k  \le 0, \quad  \forall k\in [K],  \notag  \\
  y,\lambda &\geq 0. \notag
\end{alignat}}
\end{minipage}

Note that the row-sparsity in $(P')$ is $d=K+1$ which is incomparable to $m$.  We obtain a solution to~\eqref{lp:prod-P} by setting $x=Cz$ from any solution $(\ou,z)$ to $(P')$ where $C_{m\times K}$ has $k^{th}$ column $(c_{k1},\cdots,c_{km})$. We can apply our algorithm to the convex covering problem~$(P')$ as $g'(\lambda)=\frac1p \sum_{k=1}^K \lambda_k^p$ satisfies Assumption~\ref{asm:g-mono}. This algorithm maintains monotone feasible solutions $(\ou,z)$ to $(P')$ and $(y,\lambda)$ to $(D')$. However, to solve~\eqref{lp:prod-D} online we need to maintain variables $(y,\mu)$ which is different from the variables $(y,\lambda)$ in~$(D')$. We maintain $y$ in~\eqref{lp:prod-D} to be the same as that in~$(D')$. We set the production quantities $\mu_j=\sum_{i=1}^n \sum_{T\in \cS_i} \mathbf{1}_{j\in T}\cdot y_{iT}$ so that all constraints in~\eqref{lp:prod-D} are satisfied. Note that the dual variables $y$ (allocations) and $\mu$ (production quantities) are monotone increasing- so this is a valid online algorithm.  In order to bound the objective in~\eqref{lp:prod-D} we use the feasible solution $(y,\lambda)$ to $(D')$. Note that for all $k\in [K]$:
$$c_k^T \mu \quad = \quad \sum_{j=1}^m c_{kj} \cdot \mu_j \quad = \quad \sum_{j=1}^m c_{kj} \sum_{i=1}^n \sum_{T\in \cS_i} \mathbf{1}_{j\in T}\cdot y_{iT} \quad = \quad \sum_{i=1}^n \sum_{T\in \cS_i} y_{iT} \sum_{j\in T} c_{kj} \quad \le \quad \lambda_k.$$
So the objective of $(y,\mu)$ in~\eqref{lp:prod-D} is at least that of $(y,\lambda)$ in $(D')$. Our general framework then implies a competitive ratio for the fractional problem of $\alpha=O(q\log\rho d)^q = O(q\log\rho)^q$ where
$$\rho \,\, = \,\, R\cdot \frac{K\cdot \max\{c_{kj} : k\in[K],j\in[m]\} }{\min\{c_{kj}: k\in[K],j\in[m]\}}.$$
Above $R$ is the maximum-to-minimum ratio of valuations, and recall $d\le K+1$.

Combined with Theorem~\ref{thm:prod-integral} ($\epsilon=\frac12$), we obtain:
\begin{cor}
There is a randomized online algorithm for PMPC with cost function~\eqref{eq:prod-cost-eg2} for $p>1$ that achieves expected profit at least $(1-\frac1m) \frac{\OPT}{O(q\log \rho)^q} - g(O(\log m)\cdot \mathbf{1})$.
\end{cor}
Here $g(O(\log m)\cdot \mathbf{1})\le K\cdot O\left(m\log m\cdot c_{max}\right)^p$ where $c_{max}$ is the maximum entry in $c_{kj}s$.

\bibliographystyle{alpha}
{\small \bibliography{covering-pd}}

\end{document}